\documentclass{article}
\usepackage{ijcai21}

\usepackage{multirow}
\usepackage{times}
\usepackage{soul}
\usepackage{url}
\usepackage[hidelinks]{hyperref}
\usepackage[utf8]{inputenc}
\usepackage[small]{caption}
\usepackage{graphicx}
\usepackage{amsmath}
\usepackage{amsthm}
\usepackage{booktabs}
\usepackage{tikz}
\usetikzlibrary{positioning,shapes.misc}
\usepackage{etoolbox}
\usepackage{natbib}
\usepackage{amssymb,graphicx,caption,subcaption,xcolor}
\usepackage{caption}
\usepackage{mathtools}
\usepackage{bm}
\usepackage{mathrsfs,color,enumerate}
\usepackage{comment}
\usepackage[ruled,vlined,linesnumbered]{algorithm2e}
\SetInd{0.4em}{0.5em}
\SetKw{Continue}{continue}
\SetKw{Break}{break}
\SetKwRepeat{Do}{do}{while}

\newtheorem{theorem}{Theorem}[section]
\newtheorem{proposition}[theorem]{Proposition}
\newtheorem{lemma}[theorem]{Lemma}

\theoremstyle{definition}

\newtheorem{example}[theorem]{Example}

\DeclareMathOperator*{\argmin}{arg\,min}
\DeclareMathOperator{\poly}{poly}
\newcommand{\cT}{\mathcal{T}}
\newcommand{\cC}{\mathcal{C}}
\newcommand{\ot}{\leftarrow}

\renewcommand{\d}{\mathrm{d}}
\renewcommand{\phi}{\varphi}
\newcommand{\EF}{\mathtt{EF}}
\newcommand{\SSS}{\mathtt{SSS}}
\newcommand{\WSS}{\mathtt{WSS}}

\newcounter{Bew1}
\newcounter{Bew2}

\title{Fair Ride Allocation on a Line}

\author{
Yuki Amano$^1$
\and
Ayumi Igarashi$^2$
\and
Yasushi Kawase$^{3}$
\and
Kazuhisa Makino$^1$
\and
Hirotaka Ono$^4$
\affiliations
$^1$Kyoto University\\
$^2$National Institute of Informatics\\
$^3$University of Tokyo\\
$^4$Nagoya University
\emails
$\{$ukiamano, makino$\}$@kurims.kyoto-u.ac.jp,
ayumi\_igarashi@nii.ac.jp,
kawase@mist.i.u-tokyo.ac.jp,
ono@nagoya-u.jp
}

\begin{document}

\maketitle
\begin{abstract}
The airport game is a classical and well-known model of fair cost-sharing for a single facility among multiple agents. 
This paper extends it to the so-called assignment setting, that is, for multiple facilities and agents, each agent chooses a facility to use and shares the cost with the other agents. Such a situation can be often seen in sharing economy, such as sharing fees for office desks among workers, taxis among customers of possibly different destinations on a line, and so on. 
Our model is regarded as a coalition formation game based on the fair cost-sharing of the airport game; 
we call our model \emph{a fair ride allocation on a line}.   
As criteria of solution concepts, we incorporate Nash stability and envy-freeness into our setting. 
We show that a Nash-stable feasible allocation that minimizes the social cost of agents can be computed efficiently if a feasible allocation exists. 
For envy-freeness, we provide several structural properties of envy-free allocations. 
Based on these, we design efficient algorithms for finding an envy-free allocation when at least one of  (1) the number of facilities, (2) the capacity of facilities, and (3) the number of agent types, is small. Moreover, we show that a consecutive envy-free allocation can be computed in polynomial time. 
On the negative front, we show the NP-hardness of determining the existence of an allocation under two relaxed envy-free concepts.
\end{abstract}

\section{Introduction}\label{sec:intro}
Imagine a group of university students, each of whom would like to take a taxi to her/his own destination. For example, Alice may want to directly go back home while Bob prefers to go to the downtown to meet with friends. Each of students may ride a taxi alone, or they may share a ride and split into multiple groups to benefit from sharing the cost. It is then natural to ask two problems: how to form coalitions and how to fairly divide the fee.

%airport problems
Many relevant aspects of the second problem have been studied in a classical model of the {\em airport problem}, introduced by \citet{LO73}. In the airport problem, agents are linearly ordered by their demands for a facility, and the cost of using the facility is determined by the agent who requires the largest demand. In the context of sharing a taxi, the total cost charged to a shared taxi is determined by the last agent who drops off from the taxi. While the problem originally refers to an application of the runway cost division, it covers a variety of real-life examples, 
e.g., the cost-sharing of a shared meeting room over time and an irrigation ditch; see \citet{Thomson}. In all these examples, the common property is their linear structure of the agents' demands. 

The airport problem is known to be the very first successful application of the celebrated Shapley value, 
which has a simple and explicit expression despite the exponential nature of its definition. Indeed, \citet{LO73} showed that the \emph{sequential equal contributions rule}, which applies equal division to each segment separately,  coincides with the Shapley value, and thus is the unique efficient solution that satisfies the basic desideratum of `equal treatment of equals' together with several other desirable properties, e.g., if two agents in the same group have exactly the same contribution, they will pay the same amount of money.\footnote{This rule is in fact used to split the fare in a popular fair division website of Spliddit~\citep{Goldman2015}.} 

The basic model of the airport problem, however, does not take into account the first problem, that is, how agents should form groups. In practice, facilities to be shared have capacities; so agents need to decide not only how to divide the cost, but also how to split themselves into groups so that the resulting outcome is fair across different groups. Indeed, in the preceding example of the ride-sharing, the way agents form groups affects the amount of money each agent has to pay. 
For example, consider a simple scenario of 2 taxis with capacity 3 and 4 passengers with the same destination. One might consider that the allocation in which both taxis have 2 passengers is the unique ``fair'' solution, which is indeed true with respect to \emph{envy-freeness}, though it is not with respect to \emph{Nash stability} as seen later. 
In a more complex scenario, how can we allocate passengers to taxis fairly? Which criterion of justice can we guarantee?

Envy-freeness is one of the most natural notions of fairness \citep{Foley}: 
if we select an outcome that is envy-free, no agent can replace someone else to reduce her/his cost. The notion of envies enables interpersonal comparison of utilities when agents have different needs. Another relevant criterion of justice is the notion of stability (e.g., Nash stability and swap-stability), capturing resistance to agents' deviations. No user will justifiably complain if there is no beneficial way of allocating her to another facility or swapping a pair of agents~\citep{Foley,Bogomolnaia2002,Aziz2016,bouveret-chap}. 
Social optimality and Pareto optimality are also fundamental notions related  to efficiency. \emph{Social optimality} means that there is no alternative allocation that decreases the total cost paid by the agents, whereas  
\emph{Pareto optimaility} means that there is no alternative allocation that makes some agent better off without making any agent worse off. By definition, social optimality implies Pareto optimality. %, and thus the former is a stronger notion than the latter. 
%More detailed background on these notions can be found in the supplementary material. 

\smallskip
\noindent
{\bf Our contribution}
In this paper, we extend the classical model of airport problems to the so-called assignment setting, that is, for multiple taxis and agents, each agent chooses a taxi to ride and shares the fee with the other agents riding the taxi together.
In our setting, agents would like to travel from a common starting point to their own destinations, represented by points on a line, by multiple taxis, and have to share the cost of the travel. %We consider a situation where agents start riding at same point, but the destinations can vary. 
The total cost charged to passengers for each taxi is determined by the distance between the starting point and the furthest dropping point, and is shared by the agents taking it based on the Shapley value. 
%[TODO: Add applications]
Since our model is a natural generalization of the airport game, it has potential applications such as shared office rooms; see \citet{Thomson}. 
If we restrict our attention to  fair ride allocation, the setting ``on a line'' appears a bit restrictive, and it is desirable to generalize it to more general metric cases.
However, we would like to mention that our setting is  the most fundamental case study of fair ride allocation to be investigated, and can be applied in various situations such as 
traveling to the destinations along a highway and boat travelings on a river.

%We consider the payment rule according to the sequential equal contributions rule. 
%in our model, the Shapley value coincides with the rule where each agent who still rides a taxi evenly pays the cost for each interval. 
We formulate the notions of stability and fairness including envy-freeness and Nash stability, inspired from hedonic coalition formation games and resource allocation problems, and study the existence and complexity of allocations satisfying such properties. 
%Our central technical contributions are several algorithms for finding fair and stable solutions. 

We first present basic relationships among the solution concepts.  
Concerning stability and efficiency, we show that there always exists a feasible allocation that simultaneously satisfies Nash stability, swap-stability, and social optimality, if 
a given instance contains a feasible allocation. Moreover, such an allocation can be computed in linear time by a simple backward greedy strategy. 
This contrasts to the standard results of hedonic games in two respects. First, a stable outcome does not necessarily exist in the general setting \citep{Aziz2016}.
%, we show that an instance of our problem admits an outcome that satisfies some stability requirements, such as Nash stability and swap stability. 
Second, %we are able to ensure both 
efficiency and stability are in general incompatible except for some restricted classes of games \citep{Bogomolnaia2002,Barrot2019}.

%Turning our attention to stability notions, we show that an allocation that simultaneously minimizes the total cost and satisfies several stability requirements always exists and 

For envy-freeness, there is a simple example with no envy-free feasible allocation: when 3 agents with the same destination split into 2 taxis with capacity $1$ and $2$ each, the agent who becomes alone will envy others. 
We provide three structural properties of envy-free allocations: \emph{monotonicity}, \emph{split property} and \emph{locality}. 
Based on these, we design efficient algorithms for finding an envy-free feasible allocation when at least one of (a) the number of taxis, (b) the capacity of each taxi, or (c) the number of agent types, is small. 
More precisely, in case (a), we show that the locality provides a greedy algorithm for finding an envy-free feasible allocation under a certain condition, which implies that an envy-free feasible allocation can be computed in $O(n^{3k+2})$ time, where $k$ is the number of taxis. 
%efficiently when we have a constant number of taxis. 
In case (b), we focus on the setting when the capacity of each taxi is bounded by four, where we utilize an enhanced version of split property. By combining it with the locality, we construct an $O(n^6)$-time greedy algorithm for envy-free feasible allocations.  
In case (c), that is, when the number $p$ of types is small, by utilizing the monotonicity and the split property, we first enumerate all possible `shapes' of envy-free allocations, and then compute %whether there is a size vector of coalitions that makes a given shape 
an envy-free feasible allocation in $O(p^pn^4)$ time by exploring semi-lattice structure of size vectors consistent with a given shape; a similar phenomenon has been observed in many other contexts of resource allocation (see, e.g. \citet{SunYang2003}). 
%We will also exploit the fact that the set of size vectors consistent with a given shape of an envy-free outcome enjoys a lattice structure; a% 
Note that the algorithm is FPT with respect to $p$.

We also show that one can compute an envy-free allocation that is consecutive with respect to agents' destinations by only looking at the envy between consecutive agents in $O(kn^3)$ time. 
%Moreover, we show that a consecutive envy-free allocation can be computed in polynomial time. 
As a negative side, we show that it is NP-hard to determine the existence of an allocation under two relaxed envy-free concepts.
%The related work section, omitted proofs, and examples are deferred to the supplementary material. 

\subsection{Related Work}
%Cooperative games
The problem of fairly dividing the cost among multiple agents has been long studied in the context of cooperative games with transferable utilities; we refer the reader to the book of \citet{Chalkiadakis2011} for an overview.  Following the seminal work of \citet{Shapley1953}, a number of researchers have investigated the axiomatic property of the Shapley value as well as its applications to real-life problems. 
%The axiomatic studies of the Shapley value have resulted in a variety of applications, including voting \cite{ShapleyShubik1954}, network analysis \cite{Gomez2003}, facility location \cite{BenPorat2017}, and machine learning \cite{Datta2016}, to name a few. 
\citet{LO73} analyzed the property of the Shapley value when the cost of each subset of agents is given by the maximum cost associated with the agents in that subset.
The work of \citet{Chun2017} further studied the strategic process in which agents divide the cost of the resource, showing that the division by the Shapley value is indeed a unique subgame perfect Nash equilibrium under a natural three-stage protocol.

%Congestion games
Our work is similar in spirit to the complexity study of congestion games~\citep{RN90,MS96}. In fact, without capacity constraints, it is not difficult to see that the fair ride-sharing problem can be formulated as a congestion game. The fairness notions, including envy-freeness in particular, have been well-explored in the fair division literature. Although much of the focus is on the resource allocation among individuals, several recent papers study the fair division problem among groups \citep{Kyropoulou2019,Segal-Halevi2019}. Our work is different from theirs in that agents' utilities depend not only on allocated resources, but also on the group structure. 

%Hedonic games
%Note that our model is different from the model of hedonic games due to the presence of constraints on the number of groups and the size of each group. 
In the context of hedonic coalition formation games, e.g., \citet{Bogomolnaia2002,Aziz2016,Barrot2019,Bodlaender2020}, there exists a rich body of literature studying fairness and stability. In hedonic games, agents have preferences over coalitions to which they belong, and the goal is to find a partition of agents into disjoint coalitions. While the standard model of hedonic games is too general to accommodate positive results (see \citet{Dominik}), much of the literature considers subclasses of hedonic games where desirable outcomes can be achieved. 
For example, \citet{Barrot2019} studied the compatibility between fairness and stability requirements, showing that top responsive games always admit an envy-free, individually stable, and Pareto optimal partition. 
%In the context of hedonic coalition formation games \cite{Bogomolnaia2002,Aziz2016,Barrot2019,Bodlaender2020}, the formation of groups satisfying such criteria have been extensively studied. However, the model of 

%Ridesharing
Finally, our work is related to the growing literature on ride-sharing problem \citep{Santi2014,Ashlagi2019,Pavone2012,Zhang2016,Banerjee2018,Alonso-Mora462,Chun2017,Goldman2015}. \citet{Santi2014} empirically showed a large portion of taxi trips in New York City can be shared while keeping passengers' prolonged travel time low. Motivated by an application to the ride-sharing platform, \citet{Ashlagi2019} considered the problem of matching passengers for sharing rides in an online fashion. They did not, however, study the fairness perspective of the resulting matching.

\section{Model}
For a positive integer $s \in \mathbb{Z}_{>0}$, we write $[s]=\{1,2,\ldots,s\}$. For a set $T$ and an element $a$, we may write $T+a=T \cup \{a\}$ and $T-a=T \setminus \{a\}$. 
In our setting, there are a finite set of {\em agents}, denoted by $A=[n]$, and a finite set of $k$ {\em taxis}. The nonempty subsets of agents are referred to as {\em coalitions}. 
Each agent $a \in A$ is endowed with a destination $x_a \in \mathbb{R}_{>0}$, which is called the \emph{destination type} (or shortly {\em type}) of agent $a$. 
We assume that the agents ride a taxi at the same initial location of the point $0$ and they are sorted in nondecreasing order of their destinations, i.e., $x_1\le x_2\le \dots\le x_n$.
Each taxi $i \in [k]$ has a quota $q_i$ representing its capacity,  where $q_1\ge q_2\ge \dots\ge q_k~(> 0)$ is assumed. 
An \emph{allocation} $\cT=(T_1,\dots,T_\ell)$ is an ordered partition of $A$, %, i.e., (i) $\bigcup_{i\in[k]}T_i=A$ and (ii) $T_i\cap T_j=\emptyset$ for any distinct $i,j\in[k]$. By abuse of notation, we write $T\in\cT$ to denote $T\in\{T_1,\dots,T_k\}$.
and is called \emph{feasible} if $\ell\leq k$ and $|T_i|\le q_i$ for all $i\in[\ell]$.
%We use notations $T_{<s}$, $T_{=s}$, and $T_{>s}$ to denote the set of agents with type smaller than $s$, %(i.e., $\{a\in T\mid x_a<s\}$), 
%equal to $s$, %(i.e., $\{a\in T\mid x_a=s\}$), 
%and larger than $s$, %(i.e., $\{a\in T\mid x_a>s\}$), 
%respectively. 
%Also, for a set of types $C$, we will write $T_{\in C}$ to denote the set of agents of a type in $C$ (i.e., $\{a\in T\mid x_a\in C\}$). => envyfree sectionに移動
Given a monotone nondecreasing function $f\colon \mathbb{R}_{>0}\to\mathbb{R}_{>0}$, 
the {\em cost} charged to agents $T_i$ is the value of $f$ in the furthest destination $\max_{a \in T_i}f(x_a)$ if $|T_i|\leq q_i$, and $\infty$ otherwise. 
The cost has to be divided among the agents in $T_i$.
Without loss of generality, we assume that the cost charged to $T_i$ is simply the distance of the furthest destination if $|T_i| \leq q_i$, i.e., $f$ is the identity function.
In other words, we may regard that $x_a$ is the cost itself instead of the distance.
%\footnote{We here emphasize that all the results for the distance cost are applicable to the case in which  the cost is given by a monotone nondecreasing and polynomially computable function $f$ of the furthest destination, by replacing  the destination of $a$ by $f(x_a)$, because $\max_{a\in T}f(x_a)=f(\max_{a \in T}x_a)$.}
%A natural question arises as to how the cost of each taxi should be divided. 
Among several payment rules of cooperative games, we consider a scenario where agents divide the cost using the well-known \emph{Shapley value}~\citep{Shapley1953}, which, in our setting, coincides with the following specific function. 

For each subset $T$ of agents and $s \in \mathbb{R}_{>0}$, we denote by $n_T(s)$ the number of agents $a$ in $T$ whose destinations $x_a$ is at least  $s$, i.e., $n_T(s)\coloneqq \bigl|\{a\in T\mid x_a\ge s\}\bigr|$.
For each coalition $T \subseteq A$ and positive real $x\in\mathbb{R}_{>0}$, we define
\begin{align*}
\phi(T,x)=\int_0^x \frac{\d{r}}{n_{T}(r)},
\end{align*}
where we define $\phi(T,x)=\infty$ if $n_{T}(x)=0$. 
For an allocation $\cT$ and a coalition $T_i \in \cT$, the {\em cost} of agent $a\in T_i$ is defined as $\Phi_{\cT}(a)\coloneqq \phi_i(T_i,x_a)$ where
\begin{align*}
\phi_i(T_i,x)=\begin{cases}
\phi(T_i,x) &\text{if }|T_i|\le q_i,\\
\infty    &\text{if }|T_i|>q_i.
\end{cases}
\end{align*}

It is not difficult to verify that 
%the payment rule $\phi_i$ %satisfies {\em efficiency}, i.e., 
the sum of the payments in $T_i$ is equal to the  cost of taxi $i$. Namely, if  $|T_i| \leq q_i$, we have  $\sum_{b\in T_i}\phi_i(T_i,x_{b})=\max_{a\in T_i}x_a$.
On the other hand, if  $|T_i| > q_i$, all agents in $T_i$ pay $\infty$ whose sum is equal to $\infty$ (i.e., the cost of taxi $i$). 
The following proposition formally states that the payment rule for each taxi coincides with the Shapley value. %; the proof is deferred to the appendix. 
We note that while \citet{LO73} presented a similar formulation of the Shapley value for airport games, our model is slightly different from theirs with the presence of capacity constraints. 

\begin{proposition}\label{prop:Shapley}
The payment rule $\phi_i$ is the Shapley value. 
\end{proposition}
\begin{proof}
For a given positive integer $q$, 
let $c:2^A \rightarrow \mathbb{R}$ be a cost function defined by  $c(T)=0$ if $T=\emptyset$, $\max_{a \in T}x_a$ if  $1\leq |T|\leq q$, and $\infty$ if $|T|> q$. 
Here we regard $c$ as a monotone nondecreasing function, i.e., $c(T)\geq c(S)$ for $T \supseteq S$. 
Let $T=\{a_1,\dots,a_t\}$ such that $x_{a_1}\le x_{a_2}\le\dots\le x_{a_t}$, and let $a=a_i$.
We denote by $\Pi$ the set of permutations $\pi\colon T \rightarrow [t]$. 
For a permutation $\pi \in \Pi$, we denote  
 \[
 S_{\pi}(a)=\{\, b \in T \mid \pi(b) \leq \pi(a) \}. 
 \]
Recall the definition of the Shapley value, i.e., the amount agent $a$ has to pay in the game $(T,c)$ is given by
\begin{equation}
\label{eq-defshaply1e}
\frac{1}{t!}\sum_{\pi \in \Pi}\bigl(c(S_\pi(a))-c(S_\pi(a)-a)\bigr). 
\end{equation}
If $t > q$, then there exists a permutation $\pi$ such that 
 $|S_{\pi}(a)|=q+1$ and $|S_\pi(a)-a|=q$. 
This implies that (\ref{eq-defshaply1e}) is equal to $\infty$, which shows that our payment rule is the Shapley value. 
On the other hand, if $t\leq q$,  then by introducing $x_{a_0}=0$, we have
\begin{align*}
&\sum_{\pi \in \Pi}\bigl(c(S_\pi(a))-c(S_\pi(a)-a)\bigr)\\
&=\sum_{\pi \in \Pi} \sum_{j=1}^i(x_{a_j}-x_{a_{j-1}})\bm{1}_{S_{\pi}(a) \cap \{a_j,\dots,a_{t}\}=\{a\}}\\
&=\sum_{j=1}^i(x_{a_j}-x_{a_{j-1}}) \sum_{\pi \in \Pi} \bm{1}_{S_{\pi}(a) \cap \{a_j,\dots,a_{t}\}=\{a\}}\\
&=\sum_{j=1}^i (x_{a_j}-x_{a_{j-1}}) \frac{t!}{t-j+1},
\end{align*}
Here $\bm{1}_{S_{\pi}(a) \cap \{a_j,\dots,a_{t}\}=\{a\}}$
denotes the 0-1 function that takes one if and only if 
agent $a$ appears first at $\pi$ among agents in $\{a_j,\dots,a_t\}$.
%where the third equality holds because $\bm{1}_{S_{\pi}(a) \cap \{a_j,\dots,a_{t}\}=\{a\}}$ takes value $1$ if and only if agent $a_i$ appears first at $\pi$ among agents $a_j,a_{j+1},\dots,a_t$.
Thus, we have 
\begin{align*}
&\frac{1}{t!}\sum_{\pi \in \Pi}\bigl(c(S_\pi(a))-c(S_\pi(a)-a)\bigr)\\
&=\sum_{j=1}^i\frac{x_{a_j}-x_{a_{j-1}}}{t-j+1} \\
&=\int_0^{x_a}\frac{\d{r}}{n_T(r)}=\phi(T,a). \qedhere
\end{align*}
\end{proof}

% To illustrate our payment rule, consider the following example. Here, we  use a succinct notation to specify examples. 
% %An instance will be denoted as a single arrow where the black circles on each arrow will denote the set of agents who drop off at the same destination. 
% An allocation $\cT$ is written as a set of arrows where the arrows correspond to coalitions $T \in \cT$ and the black circles on the arrow $T$ denote the set of destinations of the agents in $T$. 

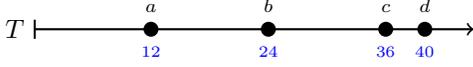
\begin{figure}[ht]
\centering
\begin{tikzpicture}[xscale=.13]
\draw[|->,thick] (0,1)--(45,1);
%\draw[|->,thick] (0,0)--(45,0);
\node[left] at (0,1) {$T$};
%\node[left] at (0,0) {$T_2$};
\node[fill=black, shape=circle, inner sep=2pt, label={[font=\scriptsize]above:{$a$}}, label={[font=\tiny,blue]below:$12$}] at (12,1) {};
\node[fill=black, shape=circle, inner sep=2pt, label={[font=\scriptsize]above:{$b$}}, label={[font=\tiny,blue]below:$24$}] at (24,1) {};
\node[fill=black, shape=circle, inner sep=2pt, label={[font=\scriptsize]above:{$c$}}, label={[font=\tiny,blue]below:$36$}] at (36,1) {};
\node[fill=black, shape=circle, inner sep=2pt, label={[font=\scriptsize]above:{$d$}}, label={[font=\tiny,blue]below:$40$}] at (40,1) {};
\end{tikzpicture}
\caption{The allocation in Example~\ref{ex:simple1}}\label{fig:simple1}
\end{figure}

\begin{example}\label{ex:simple1}
Consider a taxi that forms a coalition $T$ in Fig.~\ref{fig:simple1}, i.e.,  
 agents $a$, $b$, $c$, and $d$ take one taxi together from a starting point to the points of $12$, $24$, $36$, and $40$ on a line, respectively. 
The total cost is $40$, which corresponds to the drop-off point of $d$. 
According to the payment rule,  agents $a$, $b$, $c$, and $d$ pay $3$, $7$, $13$, and $17$, respectively. 
In fact, from the starting point to the drop-off point of $a$, 
all the agents are in the taxi,  
so they equally divide the cost of $12$, which means that $a$ should pay $3$. 
Then, between the dropping points of $a$ and $b$, three agents are in the taxi, so they equally divide the cost of $24-12=12$, which results in the cost of $4$ for each of the three agents. Thus agent $b$ pays $3+4=7$. 
By repeating similar arguments, $c$ pays $7+(36-24)/2=13$, and $d$ pays $13+(40-36)=17$. %Note that this sharing satisfies the condition of Shapley values in the cooperative game theory as explained later.
\end{example}

\section{Solution concepts}\label{sec:solution}
Agents split into coalitions and use the Shapley value to divide the  cost of each coalition. 
Our goal is to find a partition of agents that satisfies  natural desiderata. 
We introduce several desirable criteria that are inspired from coalition formation games and resource allocation problems~\citep{Foley,Bogomolnaia2002,Aziz2016,bouveret-chap}.

\medskip
\noindent{\bf Fairness}: {\em Envy-freeness} requires that no agent prefers another agent. 
Formally, for an allocation $\cT$, agent $a \in T_i$ \emph{envies} $b\in T_j$ if $a$ can be made better off by replacing herself by $b$, i.e., $i\not= j$ and  $\phi_j(T_j-b+a,x_{a})<\phi_i(T_i,x_{a})$. 
An allocation $\cT$ is \emph{envy-free (EF)} if no agent envies another agent. 
Without capacity constraints, e.g., $q_1 \geq n$, 
envy-freeness can be trivially achieved by allocating all agents to a single coalition $T_1$.  Also, when the number of taxis is at least the number of agents, i.e., $k \geq n$, an allocation  that partitions the agents into the singletons is envy-free.

\medskip
\noindent{\bf Stability}: We adapt the following three definitions of stability concepts of hedonic games \citep{Bogomolnaia2002,Aziz2016,Bodlaender2020} to our setting. The first stability concepts we introduce are those that are immune to individual deviations. For an allocation $\cT$ and two distinct taxis $i,j\in[k]$, agent $a\in T_i$ has a \emph{Nash-deviation} to $T_j$ if $\phi_j(T_j+a,x_a)<\phi_i(T_i,x_a)$. By the definition of function $\phi_j$, no agent $a$  has a Nash-deviation to $T_j$ if adding $a$ to $T_j$ violates the capacity constraint, i.e., $|T_j|\geq q_j$. An allocation $\cT$ is called \emph{Nash stable (NS)} if no agent has a Nash deviation. 
%\begin{itemize}
%\item \emph{Nash stable (NS)} if no agent has a Nash deviation;
%\item \emph{contractually individually stable (CIS)} if for each $T_i \in \cT$ and $a \in T_i$, agent $a$ does not have a Nash deviation or the cost of some agent $a' \in T_i$ increases if $a$ leaves $T_i$, i.e., $\phi_i(T_i-a,x_{a'}) > \phi_i(T_i,x_{a'})$. 
%\end{itemize}
%In our model, the definition of CIS only requires that no agent $a$ in a singleton coalition has a Nash deviation since a deviation of an agent necessarily increases the costs of the other agents in his coalition. Symmetrically, the notion of individual stability in hedonic games is equivalent to Nash stability under our setting, as a deviation to another group decreases the cost of agents in the deviating coalition. 

We  also consider stability notions that capture resistance to swap deviations. For an allocation $\cT$, agent $a\in T_i$ \emph{can replace} $b\in T_j$ if $i=j$ or   $\phi_j(T_j-b+a,x_{a}) \leq \phi_i(T_i,x_{a})$ \citep{Barrot2019,Nguyen2016}.
An allocation $\cT$ is
\begin{itemize}
\item \emph{weakly swap-stable (WSS)} if there is no pair of agents $a$ and $b$ %in distinct coalitions
such that $a$ and $b$ envy each other;
\item \emph{strongly swap-stable (SSS)} if there is no pair of agents $a$ and $b$ %in distinct coalitions  %$(i\ne j)$ 
such that $a$ envies $b$ and $b$ can replace $a$.
\end{itemize}

\medskip 
\noindent{\bf Efficiency}: Besides fairness and stability, another important property of allocation is {\em efficiency}. 
%For two feasible allocations $\cT'$ and $\cT$, we say that $\cT'$ {\em strictly Pareto dominates} $\cT$ if $\Phi_{\cT}(a)>\Phi_{\cT'}(a)$ for all $a\in A$; $\cT'$ {\em weakly Pareto dominates} $\cT$ if  $\Phi_{\cT}(a)\ge \Phi_{\cT'}(a)$ for all $a\in A$ with $\Phi_{\cT}(a)>\Phi_{\cT'}(a)$ for at least one $a\in A$. A feasible allocation $\cT$ is \emph{weakly Pareto optimal (WPO)} if there is no feasible allocation $\cT'$ that strictly Pareto dominates $\cT$; it is \emph{strongly Pareto optimal (SPO)} if there is no feasible allocation $\cT'$ that weakly Pareto dominates $\cT$.
%$\Phi_{\cT}>\Phi_{\cT'}$. 
The {\em total cost} of an allocation $\cT$ is defined as $\sum_{T \in \cT}\sum_{a \in T}\phi(T,x_a)$. 
Note that the total cost of a feasible allocation $\cT$ is equal to $\sum_{T\in\cT:\,T\ne\emptyset}\max_{a\in T}x_a$.
A feasible allocation $\cT$ is \emph{social optimal (SO)} if it minimizes the total cost over all feasible allocations. %$\sum_{a\in A}\Phi_{\cT}(a)$ 

\smallskip

In our game, we have the following containment relations among these classes of outcomes:
\begin{align}
\EF \subsetneq \SSS \subsetneq \WSS. \label{eq:relations}
\end{align}
Here, $\EF$ is defined to be the set of envy-free feasible allocations, and the other symbols are defined analogously. It is not difficult to see that the relationships with equality hold by the definitions of the concepts. To show proper inclusion, we give some examples. Moreover, %other than the relationships in \eqref{eq:relations}, any two concepts (of $\EF$, $\SSS$, $\WSS$, $\NS$, and $\SO$) are incomparable; see examples in the supplementary material. 
we show below that any two concepts with no containment relationships in \eqref{eq:relations} are incomparable. 
Namely,  there are instances with  feasible allocations that are 
(i) SO  and NS, but not WSS,  
(ii) NS and EF but not SO, and  
(iii) SO and EF but not NS, 
where they are respectively given in Examples~\ref{ex:SONSnotWSS}, \ref{ex:NSEFnotWPO},  and \ref{ex:SOEFnotNS}.  
In addition, we show that all the inclusions in  \eqref{eq:relations} are proper 
by providing the  examples with feasible allocations that are 
(a) SSS but not EF and (b) WSS but not SSS, where they are respectively given in   
Examples~\ref{ex:SPOnotSO} and \ref{ex:WSSnotSSS}.

\begin{example}\label{ex:SONSnotWSS}
Consider an instance where $n=9$, $k=2$, $q_1=5$, $q_2=4$,  $x_1=1,x_2=x_3=2$, and  $x_4=\dots=x_9=4$.
A feasible allocation $\cT=(\{2,3,7,8,9\},\{1,4,5,6\})$ in Fig.~\ref{fig:SONS} is socially optimal and Nash stable.
However, agents $1$ and $9$ envy each other, which implies that $\cT$ is not WSS.
%In addition, $\cT$ is not consecutive.
\end{example}

\begin{example}\label{ex:NSEFnotWPO}
Consider an instance where $n=4$, $k=3$, $q_1=q_2=2, q_3=4$, and $x_1=x_2=x_3=x_4=1$.
Then a feasible allocation $\cT=(\{1,2\},\{3,4\},\emptyset)$ is Nash stable and envy-free.
However, it is not socially optimal, since its total cost is larger than that of another feasible allocation $\cT'=(\emptyset,\emptyset,\{1,2,3,4\})$. 
\end{example}

\begin{example}\label{ex:SOEFnotNS}
Consider a feasible instance where $n=5$, $k=2$, $q_1=q_2=3$, $x_1=1$, $x_2=x_3=2$, and $x_4=x_5=4$.
Then a feasible allocation $\cT=(\{1,2,3\},\{4,5\})$ in  Fig.~\ref{fig:SOEFnotNS} is socially optimal.
However, agents $2$ and $3$ have a Nash deviation to $T_2$, and thus $\cT$ is not Nash stable.
\end{example}

\begin{example}\label{ex:SPOnotSO}
Consider an instance where $n=3$, $k=2$, $q_1=2$, $q_2=1$, $x_1=1$, and $x_2=x_3=2$.
Then a feasible allocation $\cT=(\{1,2\},\{3\})$ in Fig.~\ref{fig:SPOnotSO} is strongly swap-stable but not envy-free, since agent $3$ envies $1$.
\end{example}

\begin{example}\label{ex:WSSnotSSS}
Consider an instance where $n=4$, $k=2$, $q_1=q_2=2$, $x_1=x_2=1$, and $x_3=x_4=2$.
Then a feasible allocation $\cT=(\{1,3\},\{2,4\})$ in  Fig.~\ref{fig:WSSnotSSS} is weakly swap-stable but not strongly swap-stable.
\end{example}

\begin{figure}[ht]
\begin{minipage}[t]{.48\textwidth}
\centering
\begin{tikzpicture}[xscale=.9]
\draw[|->,thick] (0,0)--(5,0);
\draw[|->,thick] (0,1)--(5,1);
\node[left] at (0,1) {$T_1$};
\node[left] at (0,0) {$T_2$};
\node[fill=black, shape=circle, inner sep=2pt, label=above:{1}, label={[font=\tiny,blue]below:$1$}] at (1,0) {};
\node[fill=black, shape=circle, inner sep=2pt, label=above:{2,3}, label={[font=\tiny,blue]below:$2$}] at (2,1) {};
\node[fill=black, shape=circle, inner sep=2pt, label=above:{4,5,6}, label={[font=\tiny,blue]below:$4$}] at (4,0) {};
\node[fill=black, shape=circle, inner sep=2pt, label=above:{7,8,9}, label={[font=\tiny,blue]below:$4$}] at (4,1) {};
\end{tikzpicture}
\caption{A feasible allocation that is SO and NS but not WSS}\label{fig:SONS}
\end{minipage}\hfill
\begin{minipage}[t]{.48\textwidth}
\centering
\begin{tikzpicture}[xscale=.9]
\draw[|->,thick] (0,1)--(5,1);
\draw[|->,thick] (0,0)--(5,0);
\node[left] at (0,1) {$T_1$};
\node[left] at (0,0) {$T_2$};
\node[fill=black, shape=circle, inner sep=2pt, label=above:{1}, label={[font=\tiny,blue]below:$1$}] at (1,1) {};
\node[fill=black, shape=circle, inner sep=2pt, label=above:{2,3}, label={[font=\tiny,blue]below:$2$}] at (2,1) {};
\node[fill=black, shape=circle, inner sep=2pt, label=above:{4,5}, label={[font=\tiny,blue]below:$4$}] at (4,0) {};
\end{tikzpicture}
\caption{A feasible allocation that is SO and EF but not NS}\label{fig:SOEFnotNS}
\end{minipage}%
\hfill
\begin{minipage}[t]{.48\textwidth}
\centering
\begin{tikzpicture}[xscale=1.8]
\draw[|->,thick] (0,0)--(2.5,0);
\draw[|->,thick] (0,1)--(2.5,1);
\node[left] at (0,1) {$T_1$};
\node[left] at (0,0) {$T_2$};
\node[fill=black, shape=circle, inner sep=2pt, label=above:{1}, label={[font=\tiny,blue]below:$1$}] at (1,1) {};
\node[fill=black, shape=circle, inner sep=2pt, label=above:{2}, label={[font=\tiny,blue]below:$2$}] at (2,1) {};
\node[fill=black, shape=circle, inner sep=2pt, label=above:{3}, label={[font=\tiny,blue]below:$2$}] at (2,0) {};
\end{tikzpicture}
\caption{A feasible allocation that is SSS but not EF}\label{fig:SPOnotSO}
\end{minipage}\hfill
\begin{minipage}[t]{.48\textwidth}
\centering
\begin{tikzpicture}[xscale=1.8]
\draw[|->,thick] (0,0)--(2.5,0);
\draw[|->,thick] (0,1)--(2.5,1);
\node[left] at (0,1) {$T_1$};
\node[left] at (0,0) {$T_2$};
\node[fill=black, shape=circle, inner sep=2pt, label=above:{1}, label={[font=\tiny,blue]below:$1$}] at (1,1) {};
\node[fill=black, shape=circle, inner sep=2pt, label=above:{2}, label={[font=\tiny,blue]below:$1$}] at (1,0) {};
\node[fill=black, shape=circle, inner sep=2pt, label=above:{3}, label={[font=\tiny,blue]below:$2$}] at (2,1) {};
\node[fill=black, shape=circle, inner sep=2pt, label=above:{4}, label={[font=\tiny,blue]below:$2$}] at (2,0) {};
\end{tikzpicture}
\caption{A feasible allocation that is WSS but not SSS}\label{fig:WSSnotSSS}
\end{minipage}
\end{figure}

\section{Envy-free allocations}
In this section, we consider envy-free feasible allocations for our model.
Note that no envy-free feasible allocation exists even when a feasible allocation exists as we mentioned in Section \ref{sec:intro}.
%; see Example \ref{ex:no_envy-free} in the supplementary material.
%for the formal proof. 
%We show several fundamental properties on envy-free allocations and 
%Given that an envy-free outcome may not exist, 
We thus study the problem of deciding the existence of an envy-free feasible allocation and finding one if it exists. We identify several scenarios where an envy-free feasible allocation can be computed in polynomial time. 
We show that the problem is FPT with respect to the number of destinations, and is XP with respect to the number of taxis and the maximum capacity of a taxi.\footnote{A problem is said to be {\em fixed parameter tractable} (FPT) with respect to a parameter $p$ if each instance $I$ of this problem can be solved in time $f(p)\cdot\poly(|I|)$, and to be slice-wise polynomial (XP) with respect to $p$ if each instance $I$ of this problem can be solved in time $f(p)\cdot |I|^{g(p)}$. %; here $f(\cdot)$ and $g(\cdot)$ are computable functions that depend on $p$ only, and $\poly(\cdot)$ is an arbitrary polynomial. 
}
These restrictions are relevant in many real-life scenarios. 
For example, a taxi company may have a limited resource, in terms of both quantity and capacity. 
% it is in the introduction
It is also relevant to consider a setting where the number of destinations is small; for instance, a workshop organizer may offer a few excursion opportunities to the participants of the workshop. %generally, we expect the number of destinations to be small in many %practical applications. 
%
%For instance, a workshop organizer may offer a few excursion opportunities %to the participants of the workshop. The number of students using a school %bus of the same route may be limited. 
Furthermore, we consider consecutive envy-free feasible allocations, and show that it can be found in polynomial time. Such restrictions are intuitive to the users and hence important in practical implementation.  
As a negative remark, we show that two decision problems related to envy-free allocations are intractable. 

We start with three basic properties on envy-free allocations that will play key roles in designing efficient algorithms for the scenarios discussed in this paper. 
The first one is {\em monotonicity} of the size of coalitions in terms of the first drop-off point, which is formalized as follows.

\begin{example}\label{ex:no_envy-free}
Consider an instance where $n=4$, $k=2$, $q_1=q_2=2$,
$x_1=2$, and $x_2=x_3=x_4=4$.
We  show that no feasible allocation is envy-free. 
To see this, let $\cT=(T_1,T_2)$ be a feasible allocation. By feasibility, the capacity of each taxi must be full, i.e., $|T_1|=|T_2|=2$. Suppose without loss of generality that $T_1=\{1,2\}$ and $T_2=\{3,4\}$ in Fig.~\ref{fig:no_envy-free}. Then agent $2$ envies the agents of the same type. Indeed, she needs to pay the cost of $3$ at the current coalition while she would only pay $2$ if she were replaced by $3$ (or $4$). Hence this instance has  no envy-free feasible allocation. 
\end{example}

\begin{figure}
\centering
\begin{tikzpicture}[xscale=.9]
\draw[|->,thick] (0,1)--(5,1);
\draw[|->,thick] (0,0)--(5,0);
\node[left] at (0,1) {$T_1$};
\node[left] at (0,0) {$T_2$};
\node[fill=black, shape=circle, inner sep=2pt, label=above:{1}, label={[font=\tiny,blue]below:$2$}] at (2,1) {};
\node[fill=black, shape=circle, inner sep=2pt, label=above:{2}, label={[font=\tiny,blue]below:$4$}] at (4,1) {};
\node[fill=black, shape=circle, inner sep=2pt, label=above:{3,4}, label={[font=\tiny,blue]below:$4$}] at (4,0) {};
\end{tikzpicture}
\caption{An instance with no envy-free feasible  allocation}\label{fig:no_envy-free}
\end{figure}
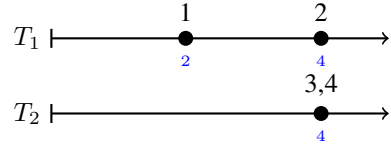

\begin{lemma}[Monotonicity lemma]\label{lemma:head_size}
For an envy-free feasible allocation $\cT$ and non-empty coalitions $T,T'\in\cT$, we have the following implications:
\begin{align}
&\min_{a\in T}x_a < \min_{a'\in T'}x_{a'} \quad\text{implies}\quad |T|\ge|T'|, \label{eq:monotone_size}\\
&\min_{a\in T}x_a=\min_{a'\in T'}x_{a'} \quad\text{implies}\quad |T|=|T'|.     \label{eq:equal_size}
\end{align}
\end{lemma}
\begin{proof}
Let $b\in\argmin_{a\in T}x_a$ and $b'\in\argmin_{a'\in T'}x_{a'}$. Suppose that $b \leq b'$ and $|T|<|T'|$.
Then $b$ envies $b'$,  because
\begin{align*}
\phi(T,x_b)=\frac{x_{b}}{|T|}>\frac{x_b}{|T'|}=\phi(T'-b'+b,x_b).
\end{align*}
Thus  $b \leq b'$ implies  $|T|\geq |T'|$, which proves \eqref{eq:monotone_size} and \eqref{eq:equal_size}.  
\end{proof}

We next show the \emph{split} property of envy-free feasible allocations.  %only the agents who drop off at the first points can ride %different taxis. 
For a coalition $T$ and a real $s$, we use notations $T_{<s}$, $T_{=s}$, and $T_{>s}$ to denote the set of agents with type smaller than $s$, %(i.e., $\{a\in T\mid x_a<s\}$), 
equal to $s$, %(i.e., $\{a\in T\mid x_a=s\}$), 
and larger than $s$, %(i.e., $\{a\in T\mid x_a>s\}$), 
respectively. 
We say that \emph{agents of type $x$ are split in an allocation $\cT$} if $\cT$ contains two distinct $T$ and $T'$ with $T_{=x}, T'_{=x}\neq\emptyset$.
The next lemma states that, the agents of type $x$ can be split in an envy-free feasible allocation only if they are the first passengers to drop off in their coalitions, and such coalitions are of the same size; further, if two taxis have an equal number of agents of split type, then no other agent rides these taxis.

An implication of the lemma is critical: we do not have to consider how to split agents of non-first drop-off points in order to see envy-free feasible allocations. 

\begin{lemma}[Split lemma]\label{lemma:split}
 If agents of type $x$ are split in an envy-free feasible allocation $\cT$, i.e., $T_{=x}, T'_{=x}\neq\emptyset$ for some distinct $T,T'\in\cT$,  then we have the following three statements: 
\begin{enumerate}
\item[{\rm(i)}] The agents of type $x$ are the first passengers to drop off in both $T$ and $T'$, i.e., $T_{<x}=T'_{<x}=\emptyset$,
\item[{\rm(ii)}] Both $T$ and $T'$ are of the same size, i.e., $|T|=|T'|$, and
\item[{\rm(iii)}] If $|T_{=x}|=|T'_{=x}|$, then $T=T_{=x}$ and $T'=T'_{=x}$.
\end{enumerate}
\end{lemma}
%\begin{proof}[Proof Sketch]
%To show $(${\rm i}$)$, if there is an agent $a \in T_{<x}$, agent $b \in T'$ of type $x$ would envy $a$, a contradiction. Hence, $T_{<x}=\emptyset$.
%By symmetry, $T'_{<x}=\emptyset$, proving $(${\rm i}$)$. This implies that $\phi(T,x)=x/|T|$ and $\phi(T',x)=x/|T'|$. Since $\phi(T,x)=\phi(T',x)$ by envy-freeness, we have $|T|=|T'|$, which proves $(${\rm ii}$)$. To see $(${\rm iii}$)$, suppose towards a contradiction that $|T_{=x}|=|T'_{=x}|$ but there is an agent in $T \cup T'$ whose destination appears strictly after $x$. Let $a^*$ be the agent with closest destination among such agents. Without loss of generality, assume $a^* \in T$. Then, it is not difficult to see that $a^*$ envies $b$, a contradiction. Hence, $|T_{=x}|=|T'_{=x}|$ implies $T=T_{=x}$ and $T'=T'_{=x}$.
%\end{proof}
\begin{proof}
Let $a\in T_{=x}$ and $b\in T'_{=x}$. As $a$ and $b$ do not envy each other, we have
\begin{align*}
\phi(T,x)=\phi(T,x_{a})\le \phi(T'-b+a,x_{a})=\phi(T',x),\\
\phi(T',x)=\phi(T',x_{b})\le \phi(T-a+b,x_{b})=\phi(T,x).
\end{align*}
Hence, $\phi(T,x)=\phi(T',x)$.

To show $(${\rm i}$)$, suppose to the contrary that $T_{<x}$ is non-empty and let $\hat{a}$ be its element.
Then, $b$ envies $\hat{a}$ because
\begin{align*}
\phi(T',x_{b})=\phi(T,x)>\phi(T-\hat{a}+b,x).
\end{align*}
Hence, $T_{<x}=\emptyset$.
By symmetry, we also have $T'_{<x}=\emptyset$, proving $(${\rm i}$)$. This implies that $\phi(T,x)=x/|T|$ and $\phi(T',x)=x/|T'|$. Since $\phi(T,x)=\phi(T',x)$ by envy-freeness, we have $|T|=|T'|$, which proves $(${\rm ii}$)$.

To see $(${\rm iii}$)$, suppose towards a contradiction that $|T_{=x}|=|T'_{=x}|$ but there is an agent in $T$ or $T'$ whose destination appears strictly after $x$, i.e., $(T\cup T' )\setminus A_{=x} \neq \emptyset$. Let $a^*$ be the agent with closest destination among such agents, i.e., $a^*\in\argmin_{a'\in (T\cup T')\setminus A_{=x}}x_{a'}$. Assume without loss of generality $a^* \in T$. Then, $a^*$ envies $b$ because 
\begin{align*}
\phi(T'-b+a^*,x_{a^*}) 
&= \frac{x}{|T'|} + \frac{(x_{a^*}-x)}{|T'|-|T'_{=x}|+1}\\
&< \frac{x}{|T|} + \frac{(x_{a^*}-x)}{|T|-|T_{=x}|}= \phi(T,x_{a^*}), 
\end{align*}
yielding a contradiction. Hence, $|T_{=x}|=|T'_{=x}|$ implies $T=T_{=x}$ and $T'=T'_{=x}$.
\end{proof}

The last property on envy-free allocations is {\em locality}, i.e.,  every agent $a$ is allocated to a taxi $T$ with minimum 
cost $\phi(T,x_a)$. 
%that charges the minimum cost with respect to the agents already %in that taxi.

\begin{lemma}[Locality lemma]\label{lemma:greedy}
For any envy-free allocation $\cT$, coalition $T\in\cT$, and agent $a\in T$, we have
\begin{align*}
\phi(T,x_a) \leq \phi(T',x_{a})
\label{eq:greedy}
\end{align*}
for all $T' \in\cT$. Furthermore, the strict inequality holds if $x_a$ is larger than the first drop off point of $T'$, i.e., $x_a > \min_{a' \in T'}x_{a'}$. 
\end{lemma}
\begin{proof}
To show the first statement, suppose towards a contradiction that there exists $T' \in\cT$ such that $\phi(T',x_a) < \phi(T,x_a)$. Then $T'$ contains an agent $a'$ with $x_{a'} \geq x_a$, 
 since otherwise $\phi(T',x_a)= \int_0^{x_a} \frac{\d{r}}{n_{T'}(r)}=\infty$.
Thus we have $\phi(T,x_a)>\phi(T',x_a)=\phi(T'-a'+a,x_a)$,
which contradicts envy-freeness of $\cT$. 

To show the second statement, assume towards a contradiction that $\cT$ contains a coalition $T'$ such that $ \min_{a' \in T'}x_{a'} < x_a$ and $\phi(T',x_a) = \phi(T,x_a)$. 
Let $a'$ be an agent in $T'$ such that $x_{a'} < x_a$. 
Then we have $\phi(T,x_a)=\phi(T',x_a) >\phi(T'-a'+a,x_a)$, 
which again contradicts envy-freeness of $\cT$. 
\end{proof}

\subsection{Constant number of taxis}
\label{sec-constanttaxix}
We start by showing that Locality lemma provides a greedy algorithm for finding  an envy-free feasible allocation if  $\argmin_{a \in T} x_a$ is known in advance for each taxi $T$. 
Especially, it implies that all envy-free feasible allocations  can be computed efficiently when we have a constant number of taxis. 

We first note that the cost of an agent $a$ in a coalition $T$ is determined by the agents of type smaller than $x_a$ and the number of agents in $T$. 
Formally, for a coalition $S\subseteq A$ and two positive reals $x$ and $\mu$, we define $\psi(S,x,\mu)$ by 
\[
\psi(S,x,\mu) := \int_0^{x} \frac{\d{r}}{n_{S}(r)+\mu-|S|}, 
\]
where $\mu\geq |S|$ is assumed. 
Then we have 
\[\phi(T,x)=\psi(T_{<x},x,|T|)\]
for any coalition  $T$ and any positive real $x$.  
Locality lemma %\ref{lemma:greedy} 
can be restated as follows.

\begin{lemma}\label{lem:greedy:psi}
For any envy-free  allocation $\cT$, coalition $T \in \cT$, and agent $a\in T$, we have
\begin{align*}
\psi(T_{<x_a},x_a,|T|) \leq \psi(T'_{<x_a},x_a,|T'|)
\end{align*}
for all $T' \in\cT$. Furthermore, the strict inequality holds if $x_a > \min_{a' \in T'}x_{a'}$. 
\end{lemma}

%By Lemma \ref{lemma:split} (i) and 
By  Lemma \ref{lem:greedy:psi}, the coalition of each agent can be determined in a greedy manner from an agent with the nearest destination, if we fix the following three parameters for each taxi $i \in [k]$: 
\begin{description}
    \item[{\rm (I)}] the number $\mu_i$ of agents who take taxi $i$,
    \item[{\rm (II)}] the first drop-off point $s_i$,
    \item[{\rm (III)}] the number $r_i$ of agents who drop off at the first point. 
\end{description}
Here we define $s_i=\infty$ if $r_i=\mu_i=0$. 
A vector $(\mu_i,s_i,r_i)_{i\in[k]}$ in $ (\mathbb{Z}_{\ge 0} \times (\mathbb{R}_{>0} \cup \{\infty\}) \times \mathbb{Z}_{\ge 0})^{[k]}$ is called a \emph{configuration} if  the following four conditions hold:
%A configuration is called {\em feasible} if 
\begin{enumerate}
    \item either $(\mu_i,s_i,r_i)=(0,\infty,0)$ %$\bigl(\mu_i=r_i=0$ and $s_i=\infty\bigr)$ 
    or 
 $\bigl(s_i\in\{x_a\mid a\in A\}$ and $1\leq r_i\le\mu_i\le q_i\bigr)$ for each $i \in [k]$, 
    \item $\sum_{i \in [k]}\mu_i=n$,  
    \item $\sum_{j \in [k]:s_j=s_i} r_i \leq  |A_{=s_i}|$   for each $i \in [k]$, and
    \item  $a \leq \sum_{i \in [k]: s_i < x_a} \mu_i +\sum_{i \in [k]: s_i =x_a} r_i$
    for each $a\in A$. 
\end{enumerate}
Here for Condition 4, we recall that for any two agents $a$ and $b$, $a<b$ implies $x_a \leq x_b$.
Note that (II) and (III) implies Condition 3. 
It is not difficult to see that $(\mu_i,s_i,r_i)_{i\in[k]}$ is a configuration if and only if there exists a feasible allocation $\cT$ that satisfies (I), (II), and (III), where such a $\cT$  is called {\em consistent with $(\mu_i,s_i,r_i)_{i\in[k]}$}. 
By definition, there exist $O(n^{3k})$ many configurations, which is polynomial when $k$ is a constant. 
%Moreover, it can be checked in polynomial time if a given vector is a configuration. 

%our greedy algorithm initializes $S_i=\emptyset$ for $i \in [k]$ and 
%allocates agents $a \in A$ to taxis $S_i$ one by one in the increasing order; 
%If $s_i=x_a$ and $|S_i| < \mu_i$ for some taxi $i$, then we allocate $a$ to arbitrary such taxi $i$ (i.e., $S_i:=S_i+a$), 
%and  otherwise,  we allocate $a$ to a taxi $i$ with minimum $\psi(S_i,x_a,\mu_i)$ among taxis with $|S_i|<\mu_i$ and $s_i \leq x_a$.

%Let us note that 
%Since we consider envy-free feasible allocations, 
%by  Lemma \ref{lemma:split} (i),  Condition 2 can be %strengthened to   
%
%%
%$2'$. $\sum_{j\in[k]:\,s_j=s_i}r_j =|A_{=s_i}|$.
%\medskip

%Given a configuration $(\mu_i,s_i,r_i)_{i\in[k]}$, our greedy algorithm can be formally described as in Algorithm~\ref{alg:const_taxi}. 

\begin{theorem}\label{thm:XP:constant:taxi}
 When the number $k$ of taxis is a constant, an envy-free feasible allocation can be found in polynomial time, if it exists.
\end{theorem}

Since all configurations can be enumerated in polynomial time if the number $k$ of taxis is a constant, 
it is sufficient to prove the following lemma.
%Lemma \ref{lemma-111} implies the following result. 
\begin{lemma}
\label{lemma-111}
Given a configuration $(\mu_i,s_i,r_i)_{i\in[k]}$, 
Algorithm \ref{alg:const_taxi} computes in polynomial time an envy-free feasible allocation consistent with $(\mu_i,s_i,r_i)_{i\in[k]}$ if it exists. 
\end{lemma}
\begin{proof}
\begin{algorithm}
\caption{ }%computing an envy-free feasible allocation consistent with $(\mu_i,s_i,r_i)_{i\in[k]}$}
\label{alg:const_taxi}
% \ForEach{$(\mu_i,s_i,r_i)_{i\in [k]}$ such that $r_i \leq \mu_i \leq q_i$ for each $i \in [k]$ and $\sum_{i \in [k]}\mu_i=n$}{
%   \lIf{$|A_{=s_i}|=\sum_{j\in[k]:\,s_j=s_i}r_j$ for each $i\in[k]$}{\Continue}
%\ForEach{configuration $(\mu_i,s_i,r_i)_{i\in [k]}$}{
  Initialize $S_i \ot \emptyset$ for each $i\in [k]$\;
  \For{$a\ot 1,2,\dots,n$}{    
    \lIf{$x_a=s_i$ and $|S_i|<r_i$ for some $i$}{
      Take such an $i$ arbitrarily
    }\label{line:head}
    \lElse{
      \mbox{Pick $i$ from \hspace{-12mm}$\displaystyle\argmin_{\hspace{12mm}j\in[k]:\, s_j <x_a\land |S_j| <\mu_j}\hspace{-12mm}\psi(S_j,x_a,\mu_j)$}
    }\label{line:greedy}
    Set $S_{i}\ot S_{i}+a$\;\label{line:set}
  }
  \lIf{$(S_1,\dots,S_k)$ is envy-free}{
    \Return $(S_1,\dots,S_k)$}\label{line:check}
%}
\lElse{\Return ``No envy-free feasible allocation consistent with $(\mu_i,s_i,r_i)_{i\in[k]}$''}
\end{algorithm} 
We prove that Algorithm~\ref{alg:const_taxi} computes in polynomial time an envy-free feasible allocation consistent with $(\mu_i,s_i,r_i)_{i\in[k]}$ if it exists. 

Let us first show that line~\ref{line:set} is executed for each agent $a$, i.e, it is allocated to some taxi $i$. 
If $x_a=s_i$ holds for some taxi $i$,
then by Condition 3, the if-statement in line~\ref{line:head} must hold, implying that $i$ is chosen in the line.
Otherwise, 
 by Conditions 2 and 4,  at least one taxi $j$ satisfies  
$s_j <x_a$ and $|S_j| <\mu_j$, which implies that $i$ is chosen in line~\ref{line:greedy}. 
Thus the algorithm allocates all the agents. 

Let ${\cal S}$ denote $(S_1,\dots,S_k)$ checked in line~\ref{line:check}. 
%Note that $|S_i|< \mu_i$ holds before $S_i$ is updated in line 7. 
%Thus 
It is not difficult to see that Conditions 1 and 2 imply that ${\cal S}$ is a feasible allocation satisfying (I). 
Moreover, Conditions 3 and 4, together with the discussion above, imply that ${\cal S}$ satisfies (II), (III) and Lemma~\ref{lemma:split} (i).
Therefore ${\cal S}$ is a feasible allocation that satisfies (I), (II), (III) and Lemma~\ref{lemma:split} (i).

We finally show that each agent $a$ is properly allocated.  
Since any agent $a$ who drops off at the first drop-off point 
(i.e., $x_a=s_i$ holds for some taxi $i$) is properly allocated, 
we only consider agents $a$ of the other kind. 
If there is an envy-free feasible allocation  consistent with a given configuration $(\mu_i,s_i,r_i)_{i\in[k]}$,  
by Lemma~\ref{lem:greedy:psi}, there exists a unique taxi $i$ that minimizes $\psi(S_i,x_a,\mu_i)$ among agents $i$ with $s_i <x_a$ and $|S_i| <\mu_i$. 
This implies that $i$ is properly chosen in line~\ref{line:greedy}.

Therefore, 
it is enough to check if ${\cal S}$ is envy-free, 
since (I), (II), (III) and  Lemma \ref{lemma:split} (i) are all necessary conditions of envy-free feasible allocations. 
This completes the proof. 
\end{proof}

\subsection{Constant capacity}
We now move on to the case when the capacity of each taxi is at most four. We design a greedy algorithm based on locality property in Lemma \ref{lem:greedy:psi}. 
Recall that the greedy algorithm works, once we fix (I), (II), and (III) in Section \ref{sec-constanttaxix}. 
If the capacity of each taxi is bounded by a constant, (I) the number $\mu_i$ of agents in taxi $i$ can be easily treated,  since
we have polynomially many candidates $\mu=(\mu_1, \dots , \mu_k)$. However, it is not immediate to handle (II) and (III), i.e., how to split the agents with the first drop-off points in taxis, even if the capacity of each taxi is bounded by four. 
In this section, we have a more detailed analysis of split property. 
More precisely, we provide all possible  {\em  split patterns} of agents with same destination which are uniquely determined in the way given in Fig.~\ref{tbl:smallsplit}.% in the supplementary material.
Based on this, we design a polynomial time algorithm for computing an envy-free feasible allocation in the case. 

\begin{theorem}\label{thm:constant:capacity}
If $q_i\le 4$ for all $i\in [k]$, then an envy-free feasible allocation can be computed in polynomial time if it exists.
\end{theorem}

We first review a few properties of envy-free feasible allocations $\cT=(T_1,\dots ,T_k)$. 
By the monotonicity in Lemma~\ref{lemma:head_size} and 
 the assumption of capacity $q_1 \geq \dots \geq q_k$,  
we can assume without loss of generality that 
\begin{align}
&\min_{a\in T_1}x_a\le\dots\le\min_{a\in T_k}x_a \label{eq:wlg}\\
&|T_1|\ge\dots\ge|T_k|. \label{eq:monotone:wlg}
\end{align}
For a destination $x$, 
let $\cT_x$ denote the family of taxis with an agent of type $x$, i.e., $\cT_x=\{i\in[k]\mid (T_i)_{=x}\ne\emptyset\}$.
By (\ref{eq:wlg}) together with Split property, we can see that $\cT_x$ consists of consecutive taxis with the same number of agents. Namely,  
$\cT_x$ can be represented by 
\[
\cT_x=\{T_s, \dots, T_t \} \mbox{ for some $s$ and $t$ in $[k]$}, 
\]
and $|T|=|T'|$ holds for any $T, T \in \cT_x$. 
Thus we further assume that for any type $x$, taxis are arranged in the nondecreasing order in terms of the number of agents of type $x$, i.e.,
\begin{align}\label{eq:split:wlg}
|(T_{s})_{=x}| \ge \dots \ge |(T_{t})_{=x}|. 
\end{align}
For a type $x$,  the sequence $(|(T_{s})_{=x}|,  \dots, |(T_{t})_{=x}|)$ is  called a {\em split pattern of $x$}. 

Let us start by proving the following auxiliary lemma to derive properties of split patterns.

%ensuring that, if a coalition $T$ consists of $|T|-1$ agents of type $x$, then all the other agents of type $x$ form coalitions without any agent of the other types, and the first passengers of the other coalitions of size $|T|$ drop off at destinations smaller than $x$.

\begin{lemma}\label{lemma:split2}
For an envy-free feasible allocation $\cT$, 
let $T$ be a coalition in $\cT$ such that  $|T|-1$ agents in $T$ drop off at the first destination, i.e., $|T_{=x}|=|T|-1$ for $x=\min_{a\in T}x_a$.
Then for any $T'\in\cT$ with $T' \neq T$ and $|T'|=|T|$, either the first destination$x'$ of $T'$  is smaller than  $x$ $($i.e., $x'<x)$, or all agents in $T'$ drop off at $x$ $($i.e., $T'\subseteq A_{=x})$.
\end{lemma}
\begin{proof}
Let $T$ be a coalition in $\cT$ such that  $|T_{=x}|=|T|-1$ for $x=\min_{a\in T}x_a$, and let $a^*$ be the unique agent in $T_{>x}$.
Take any $T'\in \cT\setminus \{T\}$ with $|T'|=|T|$ and let $x'=\min_{a\in T'}x_a$.
Assume towards a contradiction that $x' \geq x$ and  $\max_{a \in T'}x_a > x$. 
Define  $x''=\min\{x_{a^*},\max_{a \in T'}x_a\}$.
By definition, we have  $x<x''\le x_{a^*}$.
We can see that $a^*$ envies every agent $a$ in $T'_{=x'}$, because
\begin{align*}
\phi(T, x_{a^*}) &= \frac{x}{|T|}+ (x_{a^*}-x)\\
&>\frac{x}{|T'|} + \frac{x''-x}{2} +(x_{a^*}-x'')\\
&\geq \phi(T' - a + a^*, x_{a^*}),
\end{align*}
a contradiction.
\end{proof}

\begin{table}
\centering
\setlength{\tabcolsep}{4pt}
\caption{Split patterns of type $x$, where $s$  denotes the first  taxi with an agent of type $x$}
%and $d$'s are nonnegative integers}
\label{tbl:smallsplit}
\begin{tabular}{|c|c||ll|}\hline%\toprule
$|T_s|$ & $|A_{=x}|$& split patterns&\\ \hline \hline 
\multirow{5}{*}{4}     & $0 \mod 4$      & $(4,4,\dots,4)$    &\rule[-1mm]{0mm}{4mm}\\\cline{2-4}%$d$ quartets\rule[0mm]{0mm}{5mm} \\
     & $1 \mod 4$      & $(4,4,\dots,4,1)$  &\rule[-1mm]{0mm}{4mm}\\\cline{2-4}%d$ quartets and one solo\\
     & $2 \mod 4$      & $(4,4,\dots,4,2)$  &\rule[-1mm]{0mm}{4mm}\\\cline{2-4}%$d$ quartets and one duo\\
     & \multirow{2}{*}{$3 \mod 4$}    & $(4,4,\dots,4,2,1)$&{\small if 
     $|T_{s+\lceil |A_{=x}|/4\rceil}| = 4$}\\
      &           & $(4,4,\dots,4,3)$  &{\small otherwise}\rule[-1mm]{0mm}{4mm}\\\hline
\multirow{3}{*}{3}     & $0 \mod 3$        & $(3,3,\dots,3)$    &\rule[-1mm]{0mm}{4mm}\\\cline{2-4}%$d$ trios \\
     & $1 \mod 3$      & $(3,3,\dots,3,1)$  &\rule[-1mm]{0mm}{4mm}\\\cline{2-4}%$d$ trios and one solo \\
     & $2 \mod 3$      & $(3,3,\dots,3,2)$  &\rule[-1mm]{0mm}{4mm}\\\hline%$d$ trios and one duo \\
\multirow{2}{*}{2}       & $0 \mod 2$        & $(2,2,\dots,2)$    &\rule[-1mm]{0mm}{4mm}\\\cline{2-4}%$d$ duos\\
    & $1 \mod 2$      & $(2,2,\dots,2,1)$  &\rule[-1mm]{0mm}{4mm}\\\hline%$d$ duos and one solo  \\
1     &      & $(1,1,\dots,1)$    &\rule[-1mm]{0mm}{4mm}\\%$d$ solos\\ 
\hline%\bottomrule
\end{tabular}
\end{table}

The next lemma states that Table~\ref{tbl:smallsplit} represents  possible split patterns of type $x$, where 
the first column represents the size of $T_s$ for the first taxi $s$  with an agent of type $x$, the second column represents the size of $A_{=x}$,  and the last column represents possible split patterns of the corresponding cases.
For example, the first row in the table says that 
$|T_s|=4$ and $|A_{=x}|=0 \mod 4$ imply that  $(4,4,\dots, 4)$ is the unique split pattern of $x$. Thus Lemma \ref{lem:unique:split}
implies that  possible split patterns of $x$
are uniquely determined by $|T_s|$, $|A_{=x}|$,  and 
$|T_{s+\lceil |A_{=x}|/4\rceil}|$.

\begin{lemma}\label{lem:unique:split}
Suppose that $q_i \leq 4$ for $i \in [k]$. 
Let $\cT$ be an envy-free feasible allocation  satisfying \eqref{eq:wlg}, \eqref{eq:monotone:wlg},  and \eqref{eq:split:wlg}. 
%If agents of type $x$ are split in $\cT$, 
Then for any type $x$,  split patterns of type $x$ 
have the form shown in Table~\ref{tbl:smallsplit}. 
\end{lemma}
 
\begin{proof}
Recall that by Lemma \ref{lemma:split} (ii) all the taxis with an agent of type $x$ contain the same number of agents, i.e., 
\begin{align}\label{eq:same:size}
|T|=|T_s| \mbox{ for all } T\in \cT_x.
\end{align}
Moreover, by Lemma \ref{lemma:split} $($iii$)$, 
for any two taxis $T, T'\in \cT_x $,  
$|T_{=x}|, |T_{=x}| < |T_s|$ implies $|T_{=x}|\not= |T'_{=x}|$, and by 
 Lemma \ref{lemma:split2}, if a taxi $T\in \cT_x$ has 
$|T_{=x}|=  |T_s|-1$, then  we have $|T'_{=x}|=|T_s|$ for any $T' \in \cT_x$ other than $T$. 
These prove  that all the rows in Table \ref{tbl:smallsplit} are correct, except for the fourth row, i.e., the case in which $|T_s|=4$ and $|A_{=x}| =3 \mod 4$. 
For example, patterns $(4,4,\dots ,4,2,2)$ and $(4,4,\dots, 4,3,1)$  are not allowed in the first row, since the first one contains $2$ twice by Lemma \ref{lemma:split} $($iii$)$, while the second one contains $3$ and $1$ by Lemma \ref{lemma:split2}. 
We thus remain to show the case in which 
$|T_s|=4$ and $|A_{=x}| =3 \mod 4$.

In this case, by Lemmas \ref{lemma:split} $($iii$)$ and \ref{lemma:split2}, we have two possible patterns 
\[
(4,\dots,4,2,1) \text{ and } (4,\dots,4,3). 
\]
Let $|A_{=x}|=4d +3$ for some nonnegative integer $d$, and 
assume towards a contradiction that $|T_{s+d+1}| = 4$ and $(4,\dots,4,3)$ is a split pattern. 
In this case $s+d$ is the last taxi with an agent of type $x$, and 
we have 
$|T_{s+d}|=|T_{s+d+1}|=4$, $|(T_{s+d})_{=x}|=3$, and $(|T_{s+d+1})_{=x}|=0$. 
This contradicts Lemma \ref{lemma:split2}, since all the agents in taxi $s+d+1$ have types larger than $x$. Thus $(4,\dots,4,2,1)$ is a possible split pattern if $|T_{s+d+1}| = 4$. 
On the other hand, if $|T_{s+d+1}| <4$,  $(4,\dots,4,3)$ is a possible split pattern, since otherwise, 
taxi $s+d+1$ contains an agent of type $x$, which contradicts
(\ref{eq:same:size}). 
%This completes the proof. 
\end{proof}

%We show that an allocation $\cT=(T_1,\dots,T_k)$ with $\cT \simeq \cT^*$
%$|T_i\cap A_{=x}|=|T_i^*\cap A_{=x}|~(\forall i\in[k],x\in\mathbb{R}_{>0})$ 

%\smallskip
%\noindent
%{\em Algorithm description}: 
In outline, our algorithm guesses the size of each coalition $T_i$ ($i\in [k]$) and greedily allocates each agent from the smallest type $x$ to the largest one. The formal description of the algorithm is given by Algorithm~\ref{alg:le4}. More precisely, let $\mathcal{M}$ be the set of $k$-tuples of integers $(\mu_1,\dots,\mu_k)$ such that
$\mu_1\ge\dots\ge\mu_k\ge 0$, 
$\sum_{i\in[k]}\mu_{i}=n$, and 
$\mu_i\le q_{i}$ for all $i\in[k]$. 
If $k \geq n$, it always has an envy-free feasible allocation, by allocating each agent to each taxi. 
Thus we assume that $k < n$.
If  $q_i\le 4$ for all $i\in [k]$, we have $|\mathcal{M}|=O(n^4)$, since each of the first $k_4$ taxis contains
four agents, each of the next $k_3$ taxis contains three agents, and so on. 
Our algorithm enumerates all the $k$-tuples in $\mathcal{M}$ in polynomial time, and for each $(\mu_1,\dots,\mu_k)\in\mathcal{M}$, applies a greedy method based on Locality property. 
Namely, 
we greedily add agents with the smallest available type $x$ to taxis $i$ with minimum cost $\psi(T_i,x,\mu_i)$.
Recall the discussion in Section \ref{sec-constanttaxix}: the greedy method does not provide an envy-free feasible allocation if multiple taxis $i$ attain the minimum cost $\psi(T_i,x,\mu_i)$.
However, by making use of Lemma \ref{lem:unique:split}, 
we can show that it works if we apply the  simple rule that chooses the smallest $i$  with minimum $\psi(T_i,x,\mu_i)$, except for the case corresponding to 
the fourth row in Table \ref{tbl:smallsplit}.

\begin{algorithm}
\caption{Polynomial-time algorithm for taxis with capacity at most $4$}\label{alg:le4}
\ForEach{$(\mu_1,\dots,\mu_k)\in\mathcal{M}$}{
  Let $T_i\ot\emptyset$ for each $i\in [k]$\;
  \For(\tcp*[h]{\!\!{\scriptsize from\! nearest \!to\! farthest}\hspace*{-5mm} }){$a\ot 1,2,\dots,n$}{
    Let $i^*$ be the smallest index $i$ that minimizes
    %$\int_0^{x_a}\frac{\d{r}}{n_{T_i}(r)+(\mu_i-|T_i|)}$
    $\psi(T_i,x_a,\mu_i)$
    among taxis $i$ with $|T_i|<\mu_i$\label{line:greedy2}\;
    \label{line:if} \If{$\mu_{i^*}=\mu_{i^*+1}=4$, $T_{i^*}=\{b,c\}$, and $x_{a}=x_{b}=x_c<x_{a+1}$}{
        Set $T_{i^*+1}\ot T_{i^*+1}+a$\ ; \label{line:exception}
    }
    \lElse{Set $T_{i^*}\ot T_{i^*}+a$} \label{line:normal}
  }
  \lIf{$(T_1,\dots,T_k)$ is envy-free}{\Return $(T_1,\dots,T_k)$}
}
\Return ``No envy-free feasible allocation''\;
\end{algorithm}

We formally show that Algorithm~\ref{alg:le4} computes an envy-free feasible allocation in polynomial time if a given instance contains such an allocation.

\begin{proof}[Proof of Theorem \ref{thm:constant:capacity}]
It is not difficult see that Algorithm~\ref{alg:le4} returns an envy-free feasible allocation if it returns an allocation.
Suppose that there exists an envy-free feasible allocation $\cT^*$ with $\mu^*_i=|T^*_i|$ for all $i\in[k]$. Without loss of generality, we assume that $\cT^*$ satisfies \eqref{eq:wlg}, \eqref{eq:monotone:wlg}, and \eqref{eq:split:wlg}. 
We show that the algorithm computes an envy-free allocation isomorphic to $\cT^*$ if the for-loop of  $(\mu_1^*,\dots ,\mu^*_k)$ is applied, which proves the correctness of the algorithm. 
We thus restrict our attention to  the for-loop of $(\mu_1^*,\dots ,\mu^*_k)$, and inductively prove that 
the partial allocation $\cT^{(j)}$ after the $j$-th iteration of $a$ is extendable to an envy-free feasible (complete) allocation isomorphic to $\cT^*$. 
Before the induction, we note that 
 allocation $\cT^{(n)}$ is feasible and satisfies \eqref{eq:monotone:wlg} by the assumption on $\cT^*$. 
Moreover, 
at any iteration of $a$, agent $a$ is allocated to the taxi which already has an agent or the first taxi with no agent, i.e., for any $j \in [n]$ and for any $i, \ell \in [k]$ with $i \leq \ell$, we have 
\begin{equation}
    \label{eq-algoproof3}
T^{(j)}_i=\emptyset \ \Longrightarrow \ T^{(j)}_\ell=\emptyset \mbox{ }, 
\end{equation}
which implies 
\begin{equation}
    \label{eq-algoproof32}
\min_{a\in T^{(j)}_1}x_a\le\dots\le\min_{a\in T^{(j)}_k}x_a 
\ \ \mbox{  for any $j \in [n]$}. 
\end{equation}
By substituting $j$ by $n$, we have \eqref{eq:wlg}, 
and \eqref{eq:split:wlg} is satisfied 
by \eqref{eq-algoproof32}, together with the choice of $i^*$ in line 4 of the algorithm. 
Let us now apply the induction. 
By Lemma \ref{lem:greedy:psi},  
it is clear that $\cT^{(1)}$ is extendable to a desired allocation. 
Assuming that $\cT^{(j)}$ is extendable to a desired allocation, we consider the $(j+1)$-th iteration of $a$.  
Let 
\[Q=\argmin\{\psi(T^{(j)}_i,x_{j+1},\mu_i)\mid |T^{(j)}_i| <\mu_i\}.  \]
If $Q$ contains a  taxi $i$ such that $T^{(j)}_i$  has an agent  of type smaller than $x_{j+1}$, 
no other taxi in $Q$ has such a property, since otherwise 
Lemma \ref{lem:greedy:psi} provides a contradiction. 
Moreover, $j+1$ must be contained in such a taxi $i$ again by Lemma \ref{lem:greedy:psi}. 
Since the algorithm chooses such a taxi $i$ by \eqref{eq-algoproof32}, 
 $\cT^{(j+1)}$ is extendable to a desired allocation.
On the other hand, If $Q$ contains no such taxi, i.e., 
a taxi $i$ in $Q$ satisfies either (i) $T^{(j)}_i$ is empty or  (2) it  consists of agents  of type $x_{j+1}$, 
then the algorithm again chooses a correct $i^*$, since it fits with possible split patterns in  Lemma \ref{lem:unique:split}.
Thus $\cT^{(j+1)}$ is extendable to a desired allocation.
This completes the induction. 

It remains to show the time complexity of the algorithm. 
Note that  $|\mathcal{M}|=O(n^4)$ and $\mathcal{M}$ is constructed in the same amount of time. 
For each $(\mu_1,\dots,\mu_k)\in\mathcal{M}$, 
the for-loop  is executed in   $O(n^2)$ time.
 Therefore, in total,  the algorithm requires $O(n^4\times n^2)=O(n^6)$ time. 
\end{proof}

By the proof above, if there exist an envy-free feasible allocation consistent with  $(\mu_1,\dots,\mu_k)\in\mathcal{M}$, then it is unique up to isomorphism.  
We also remark that the greedy algorithm above  cannot  be directly extended to the case of constant capacity, 
since split patterns are not uniquely determined, even when the maximum capacity is at most $5$.

\subsection{Small number of types}
In this section, we focus on Split lemma of envy-free allocations. We represent envy-free allocations by directed graphs $G$ together with size vectors $\lambda$. 
We provide several structural properties of $G$ and $\lambda$.
Especially, we show that $G$ and $\lambda$ define a unique envy-free allocation (up to isomorphism), $G$ is a star-forest, and $\lambda$ forms
 semi-lattice. Based on their properties, we show that an envy-free feasible allocation can be computed in FPT time with respect to the number of destination types. 
 
%Throughout, we assume that there are $p$ types of agents.%, i.e., $p=\bigl|\{x_a\mid a\in A\}\bigr|$. 
Let $V=\{x_a\mid a\in A\}$ be the set of destination types, and let $p=|V|$. 
For an allocation $\cT=(T_1,\dots,T_k)$, we define its \emph{allocation (di)graph} $G^\cT=(V,E)$ by 
\begin{align*}
E&=\bigcup_{T\in\cT}\left\{(y,z)\in V^2 \,\middle|\!
{\small\begin{tabular}{l}
$y,z\in \{x_a\mid a\in T\}$, 
$y<z$,\\
$\not\exists a\in T: y<x_a<z$
\end{tabular}}\!\!\!\!
\right\}.
\end{align*}
Namely, the allocation graph $G^\cT$ contains an directed edge $(y, z)$ if and only if an agent of type $y$ drops off just after an agent of type $z$ in some coalition $T\in\cT$. 
By definition, $G^\cT$ is acyclic because every edge is oriented from a smaller type to a larger type, i.e.,  $(y,z) \in E$ implies $y<z$. 
We assume that all graphs discussed in this section satisfy the condition.
% \begin{equation}
% \label{eq-acycle1}
% (y,z) \in E \mbox{ implies } y<z. 
% \end{equation}
%We assume that all graphs discussed in this section satisfy (\ref{eq-acycle1}). 

A graph is called a {\em star-tree} if it is a rooted (out-)tree such that all vertices except the root have out-degree at most $1$, and  a {\em star-forest} if each connected component is a star-tree. 
Then (i) in  Split lemma implies that $G^\cT$ is a star-forest.
See the allocation graph for an envy-free feasible allocation is depicted in Fig.~\ref{fig:trees}.

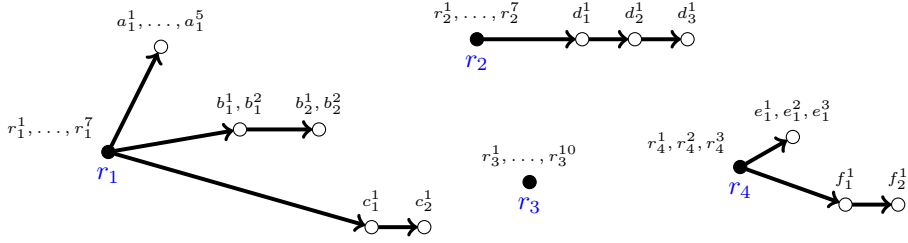
\begin{figure*}[htbp]
\centering
\begin{tikzpicture}[xscale=.7,p/.style={circle,fill=white,inner sep=0pt,outer sep=0pt,minimum size=5pt,draw}]
\node[p,fill=black,label={[blue]below:$r_1$},label={[font=\tiny]above left:$r_1^1,\dots,r_1^7$}] at (0,0) (r1) {}; 
\node[p,label={[font=\tiny]above:$a_1^1,\dots,a_1^5$}] at (1, 1.4) (a1) {}; 
\node[p,label={[font=\tiny]above:$b_1^1,b_1^2$}] at (2.5, 0.3) (b1) {}; 
\node[p,label={[font=\tiny]above:$b_2^1,b_2^2$}] at (4, 0.3) (b2) {}; 
\node[p,label={[font=\tiny]above:$c_1^1$}] at (5, -1) (c1) {}; 
\node[p,label={[font=\tiny]above:$c_2^1$}] at (6, -1) (c2) {}; 

\node[p,fill=black,label={[blue]below:$r_2$},label={[font=\tiny]above:$r_2^1,\dots,r_2^7$}] at (7,1.5) (r2) {}; 
\node[p,label={[font=\tiny]above:$d_1^1$}] at (9, 1.5) (d1) {}; 
\node[p,label={[font=\tiny]above:$d_2^1$}] at (10, 1.5) (d2) {}; 
\node[p,label={[font=\tiny]above:$d_3^1$}] at (11, 1.5) (d3) {}; 

\node[p,fill=black,label={[blue]below:$r_3$},label={[font=\tiny]above:$r_3^1,\dots,r_3^{10}$}] at (8,-.4) (r3) {}; 

\node[p,fill=black,label={[blue]below:$r_4$},label={[font=\tiny]above left:$r_4^1,r_4^2,r_4^3$}] at (12,-0.2) (r4) {}; 
\node[p,label={[font=\tiny]above:$e_1^1,e_1^2,e_1^3$}] at (13, .2) (e1) {}; 
\node[p,label={[font=\tiny]above:$f_1^1$}] at (14, -.7) (f1) {}; 
\node[p,label={[font=\tiny]above:$f_2^1$}] at (15, -.7) (f2) {}; 

\foreach \u/\v in {r1/a1,r1/b1,b1/b2,r1/c1,c1/c2,r2/d1,d1/d2,d2/d3,r4/e1,r4/f1,f1/f2}
  \draw[line width=1.5pt,->] (\u) -- (\v);
\end{tikzpicture}
\caption{An example of the allocation graph for an envy-free feasible allocation $\cT=(T_1,T_2,\ldots,T_9)$ where 
$T_1=\{r_1^1,a_1^1,a_1^2,a_1^3,a_1^4,a_1^5\}$, 
$T_2=\{r_1^2,r_1^3,b_1^1,b_1^2,b_2^1,b_2^2\}$, 
$T_3=\{r_1^4,r_1^5,r_1^6,r_1^7,c_1^1,c_2^1\}$, 
$T_4=\{r_2^1,r_2^2,d_1^1,d_2^1,d_3^1\}$, 
$T_5=\{r_2^3,r_2^4,r_2^5,r_2^6,r_2^7\}$, 
$T_6=\{r_3^1,r_3^2,r_3^3,r_3^4,r_3^5\}$, 
$T_7=\{r_3^6,r_3^7,r_3^8,r_3^9,r_3^{10}\}$, 
$T_8=\{r_4^1,e_1^1,e_1^2,e_1^3\}$, 
$T_9=\{r_4^2,r_4^3,f_1^1,f_2^1\}$. 
There are seven agents of type $r_1$ $(r^1_1,\dots,r^7_1)$, seven agents of type $r_2$ $(r^1_2,\dots,r^7_2)$, ten agents of type $r_3$ $(r^1_3,\dots,r^{10}_3)$, and three agents of type $r_4$ $(r^1_4,r^2_4,r^3_4)$.}
\label{fig:trees}
\end{figure*}

Now, we will explore the relationship between $\cT$ and $G^\cT$, implied by Split lemma. 
Formally, let $\cC=\{C_1,\dots,C_t\}$ be the family of the vertex sets of connected components in $G^\cT$. 
Let $r_j$ be the root of $C_j$, i.e.,  $r_j=\min_{x \in C_j} x$, and let $d_j$ be out-degree of $r_j$.    
We assume that the components are arranged in ascending order of the root, i.e., $r_1<\dots<r_t$.
Let $\cT_j$ be the family of  coalitions $T \in \cT$ in which all members  have types in $C_j$. 
To see this, we write $T_{\in C}$ to denote $T_{\in C}=\{a\in T\mid x_a\in C\}$ for a coalition $T$ and a set of types $C$; then $\cT_j=\{T \in \cT \mid T=T_{\in C_j}\}$. 
By definition of $G^\cT$,  $\{\cT_1, \dots , \cT_t\}$ is a partition of $\cT$.

By star-tree property of $C_j$, vertices $C_j\setminus\{r_j\}$ forms $d_j$ paths in $G^\cT$. 
Let $C_j^\ell$ $(\ell =1, \dots, d_j)$ be the vertex sets of such paths.
Then by Split lemma, we have the following three conditions: \begin{align}
&\mbox{each $T \in \cT_j$ satisfies 
either $\emptyset \not= T \subseteq A_{=r_j}$}\nonumber\\
&\hspace*{1cm}\mbox{or 
$A_{\in C_j^\ell} \subsetneq T \subseteq A_{\in C_j^\ell} \cup A_{=r_j}$ for some $\ell$.} \label{eq-max1}\\
&\mbox{$|T|=|T'|$ holds for any $T, T'\in \cT_j$, and}\label{eq-max2}\\
&\mbox{$|A_{\in C_j^\ell}|\not=|A_{\in C_j^h}|$ for any distinct $\ell,h \in [d_j]$}. \label{eq-max3}
\end{align}

By (\ref{eq-max1}), some agents of type $r_j$ form a coalition $T$ or some agents of type $r_j$ together with the agents of types in $C_j^\ell$ form a coalition. 
%See an envy-free feasible allocation depicted in Fig.~\ref{fig:trees}.
It follows from (\ref{eq-max2}) that each coalition in $\cT_j$ has the same size $\lambda_j$. 
Let us call $\lambda^\cT=(\lambda^\cT_1, \dots , \lambda^\cT_t)$ the {\em size vector} of $\cT$.  
%For instance, $\lambda^{\cT}_1=6$, $\lambda^{\cT}_2=\lambda^{\cT}_3=5$, and $\lambda^{\cT}_4=4$ in Fig.~\ref{fig:trees}. 
In summary, we have the following result as stated in Lemma \ref{lem:starforest}, where isomorphism $\simeq$ of two allocations $\cT=(T_1, \dots , T_\alpha)$ and $\cT'=(T'_1, \dots , T'_{\beta})$ is defined as follows:
for two coalitions $T$ and $T'$, we write $T\simeq T'$ to mean that $T$ and $T'$ contains the same number of agents for each type, i.e., $|T_{=y}|=|T'_{=y}|$ for all $y\in V$; for two allocations $\cT$ and $\cT'$, we write $\cT\simeq \cT'$ if $|\cT|=|\cT'|$ and there exists a permutation $\sigma\colon[\alpha]\to[\alpha]$ such that $T_i\simeq T'_{\sigma(i)}$ for all $i\in [\alpha]$. 

\begin{lemma}\label{lem:starforest}
Suppose that an allocation $\cT$ satisfies the conditions in Lemma \ref{lemma:split}. 
Then $G=G^\cT$ and  $\lambda=\lambda^\cT$ satisfy the following conditions: 
\begin{align}
&\text{$G$ is a star-forest with (\ref{eq-max3}) for any $j$ in $[t]$, and}\label{eq:splitq1}\\
&\text{for any $j$ in $[t]$, 
$\lambda_j$ is a divisor of $|A_{\in C_j}|$}\nonumber\\ 
&\hspace*{1.1cm}\text{such that $\max_{\ell\in[d_j]}|A_{\in C_j^\ell}| \leq \lambda_j \leq |A_{\in C_j}|/d_j$. }\label{eq:splitq2}
\end{align}
Conversely, if $G$ and $\lambda$ satisfy the conditions above, then there exists a unique allocation $\cT$ (up to isomorphism) satisfying $G^\cT=G$, $\lambda^\cT=\lambda$, and the conditions in Lemma~\ref{lemma:split}. 
\end{lemma}
\begin{proof}
Suppose that an allocation $\cT$  satisfies the conditions in Lemma \ref{lemma:split}.
It is not difficult to see that (\ref{eq:splitq1}) follows from the discussion above and  (\ref{eq-max3}), and  (\ref{eq:splitq2}) follows from  (\ref{eq-max1}) and (\ref{eq-max2}). 
Conversely, if $G$ and $\lambda$ satisfy  (\ref{eq:splitq1}) and  (\ref{eq:splitq2}), then we can construct a unique allocation $\cT$ up to isomorphism that satisfies  (\ref{eq-max1}), (\ref{eq-max2}), and   (\ref{eq-max3}). Thus $\cT$ satisfies  the conditions in Lemma \ref{lemma:split}.
\end{proof}

We note that a unique allocation $\cT$ in the converse statement can be computed in polynomial time if $G$ and $\lambda$ are given. 
Thus, a naive approach to find an envy-free feasible allocation is to enumerate all possible $G$ and $\lambda$, and then check if they provide a envy-free feasible allocation. 
Note that  the number of star-forests is at most $p^p$,  because the in-degree of every node is at most one.
However, we may have $n^{\Omega(p)}$ many candidates of $\lambda$, even if a star-forest $G$ is fixed in advance.  
To overcome this difficulty, we show that for a given star-forest $G$,  the size vectors $\lambda$  such that $G$ and $\lambda$ provide envy-free feasible allocations form a semi-lattice. 
More precisely, for a star-forest $G$, let $\Lambda_{G}$ denote the set of size vectors $\lambda$ such that 
$G$ and $\lambda$ provide envy-free feasible allocations. 
Then we have the following structural property of $\Lambda_G$

\begin{lemma}\label{lem:lattice}
For any star-forest $G$, $\Lambda_G$ is an upper semilattice with respect to the componentwise max operation $\vee$, i.e., $\lambda, \lambda' \in  \Lambda_G$ implies $\lambda \vee \lambda' \in \Lambda_G$
\end{lemma}
%%%%%%%%%%%%%%%%%%%%%
%\appendixproof{Proofs of Lemmas \ref{lem:lattice} and \ref{lemma:FPT_main-simple}}{
In this section, we show the following lemma, which is stronger than both  Lemmas \ref{lem:lattice} and \ref{lemma:FPT_main-simple}.  

\begin{lemma}\label{lemma:FPT_main}
Let $G$ be a star-forest, and let $\Lambda=\prod_{j\in[t]}\Lambda_j$  be a non-empty set in $\mathbb{Z}_{>0}^t$. 
If the maximum vector  $\lambda=(\max \Lambda_j)_{j\in[t]}$ does not belong to $\Lambda_G$, then there exists an index $\ell\in[t]$ such that 
\begin{equation}
\label{eq-lattice000}
\Bigl((\Lambda_{\ell}-\max\Lambda_{\ell}) \times \prod_{j\in[t]-\ell}\Lambda_j \Bigr)\cap \Lambda_G= \Lambda \cap \Lambda_G.    
\end{equation} 
In addition, such an index $\ell$ can be computed in polynomial time.
\end{lemma}

We note that Lemma \ref{lemma:FPT_main} implies semilattice property of $\Lambda_G$. 
To see this, suppose that $\Lambda_G$ is not a semilattice, i.e., there exists two size vectors $\lambda, \lambda' \in \Lambda_G$ such that $\lambda \vee \lambda' \not\in \Lambda_G$. 
Then we define $\Lambda$ by  $\Lambda_i=[(\lambda \vee \lambda')_i]$ for $i \in [t]$. 
By definition, $\lambda, \lambda' \in \Lambda$ and 
$\lambda \vee \lambda'$ is the maximum vector in $\Lambda$ such that $\lambda \vee \lambda' \not\in \Lambda_G$. 
However, no index $\ell$ satisfies (\ref{eq-lattice000}), since the right-hand side of (\ref{eq-lattice000}) contains both $\lambda, \lambda'$, while the left-hand side of (\ref{eq-lattice000}) contains at most one of them. 
%$\lambda, \lambda' \in \Lambda \cap \Lambda_G$.    
Furthermore, if a set $\Lambda$ in Lemma \ref{lemma:FPT_main} is chosen in such a way that $\Lambda \supseteq \Lambda_G$, we obtain
Lemma~\ref{lemma:FPT_main-simple}. 

In order to show Lemma \ref{lemma:FPT_main}, 
let us consider feasibility and  monotonicity of  allocations in addition to split property.

\begin{lemma}\label{lemma-monotonefeasible}
An allocation $\cT$ is feasible and  satisfies the conditions in Lemmas \ref{lemma:head_size} and \ref{lemma:split}. 
Then $\lambda=\lambda^\cT$ satisfy the following conditions. \begin{align}
&\text{$\lambda_1\ge\lambda_2\ge\dots\ge\lambda_t$}
\label{eq:splitq3}\\
&\text{$\textstyle\sum_{j\in[t]}|A_{\in C_j}|/\lambda_j\le k$, and}\label{eq:splitq4}\\
&\text{$\textstyle \lambda_j \leq q_{\eta(j)}$ for all $j \in [t]$}\label{eq:splitq5}.
\end{align}
where $\eta(j)=\sum_{r\leq j}|A_{\in C_r}|/\lambda_r$. 
Conversely, if  $G$ and $\lambda$ satisfy 
(\ref{eq:splitq1}), (\ref{eq:splitq2}), (\ref{eq:splitq3}), (\ref{eq:splitq4}), and  (\ref{eq:splitq5}), then there exists a unique feasible allocation $\cT$ (up to isomorphism) satisfying $G^\cT=G$,  $\lambda^\cT=\lambda$, and the conditions in Lemmas \ref{lemma:head_size} and \ref{lemma:split}. 
\end{lemma}

\begin{proof}
Suppose that $\cT$ is a feasible allocation satisfying the conditions in Lemmas \ref{lemma:head_size} and \ref{lemma:split}. 
By our assumption $r_1 < \dots  r_t$, (\ref{eq:monotone_size}) implies (\ref{eq:splitq3}). Note that the feasibility of $\cT$ is equivalent to two conditions (i) $|\cT|\leq k$ and (ii) capacity condition (i.e., $|T_i| \leq q_i$). 
Since $\cT_j$ uses $|A_{\in C_j}|/\lambda_j$ many taxis, (i) is equivalent to (\ref{eq:splitq4}). 
By (\ref{eq:splitq3}) and the assumption $q_1 \geq \dots q_k$, in order to check capacity condition, it is enough to consider an allocation $\cT=(T_1, \dots, T_\alpha)$ in such a way that $\cT_1$ is assigned to the first $\eta(1)$ taxis, $\cT_2$ is assigned to next $\eta(2) -\eta(1)$ taxis, and so on. 
More precisely, we have 
\[\cT_j=\{T_{\eta(j-1)+1}, \dots, T_{\eta(j)}\} \mbox{ for all } j \in [t], \]
where $\eta(0)$ is defined by $0$. 
Thus the capacity condition implies (\ref{eq:splitq5}).  Conversely, if  $G$ and $\lambda$ satisfy 
(\ref{eq:splitq1}) and (\ref{eq:splitq2}), then 
then by Lemma \ref{lemma:head_size}, there exists a unique allocation $\cT$ (up to isomorphism) satisfying $G^\cT=G$, $\lambda^\cT=\lambda$, and the conditions in Lemma \ref{lemma:split}. 
Moreover, since (\ref{eq:splitq3}), (\ref{eq:splitq4}), and  (\ref{eq:splitq5}) hold for $\lambda$, $\cT$ is feasible  and  
the conditions in Lemma \ref{lemma:head_size} are satisfied. 
\end{proof}

We are now ready to prove Lemma \ref{lemma:FPT_main}.

\begin{proof}[Proof of Lemma \ref{lemma:FPT_main}]
Let us separately consider the cases in which 
$G$ and $\lambda=(\max \Lambda_j)_{j\in[t]}$ violate 
(\ref{eq:splitq1}), (\ref{eq:splitq2}), (\ref{eq:splitq3}), (\ref{eq:splitq4}),  (\ref{eq:splitq5}), and envy-freeness of the allocation provided by them.  
\begin{itemize}
\item 
If (\ref{eq:splitq1}) or (\ref{eq:splitq4}) is violated, then by Lemmas \ref{lem:starforest} and \ref{lemma-monotonefeasible},  we have $\Lambda_G=\emptyset$. 
This implies that any index $\ell$ satisfies 
(\ref{eq-lattice000}). Thus it is polynomially computable. 

\item
If (\ref{eq:splitq2}) is violated for an index $j$, then 
$\ell=j$ satisfies 
(\ref{eq-lattice000}). Thus it is polynomially computable. 

\item
If (\ref{eq:splitq3})  is violated for an index $j$, i.e., $\lambda_{j-1} < \lambda_j$, then 
$\ell=j$ satisfies 
(\ref{eq-lattice000}). Thus it is polynomially computable. 

\item
If (\ref{eq:splitq5})  is violated for an index $j$, i.e., $\lambda_j > q_{\eta(j)}$, then we claim that  
$\ell=j$ satisfies 
(\ref{eq-lattice000}), which completes the proof of this case, since such an $\ell$ can be computed in polynomial time. 
Let $\lambda'$ be a size vector in $\Lambda$ such that $\lambda'_\ell=\lambda_\ell$, and let $\eta'(h)=\sum_{r\leq h}|A_{\in C_r}|/\lambda'_r$ for $h \in [t]$. 
Since $\lambda'\leq \lambda$ and $\lambda'_\ell=\lambda_\ell$,  
we have $\lambda'_\ell=\lambda_\ell > \eta(\ell)\geq \eta'(\ell)$, which implies the claim.  

\item 
Suppose that $G$ and $\lambda$ fulfill all the conditions above, i.e., 
$G$ and $\lambda$ provide a feasible allocation $\cT$  that satisfies the conditions in Lemmas \ref{lemma:head_size} and \ref{lemma:split}. 
Let further assume that $a \in T \,(\in \cT_h)$ envies $a' \in T'\,(\in \cT_j)$ for some $j,h \in [t]$. 
 If $j=h$, 
 then it is clear that $\ell=j\,(=h)$ satisfies 
(\ref{eq-lattice000}).
On the other hand, if $j\not= h$, 
Let $\lambda'$ be a size vector in $\Lambda$ such that $\lambda'_\ell=\lambda_\ell$ and satisfies (\ref{eq:splitq2}), (\ref{eq:splitq3}), (\ref{eq:splitq4}),  and (\ref{eq:splitq5}). 
Then $a$ still envies $a'$ in the allocation provided by $G$ and $\lambda'$. 
Thus $\ell=j$ again satisfies (\ref{eq-lattice000}).
Since envy-freeness can be checked in polynomial time, 
this completes the proof. \qedhere
\end{itemize}
\end{proof}

%%%%%%%%%%%%%%%%%%%
We here remark that $\Lambda_G$ may be empty. 
Based on this semi-lattice structure, we construct a polynomial time algorithm to compute an envy-free feasible allocation consistent with a given star-forest $G$.
Since there exists at most $p^p$ many star-forests, 
this implies an FPT algorithm (with respect to $p$) for 
computing an envy-free feasible allocation.

For a given star-forest $G$, our algorithm computes 
 the maximum vector in $\Lambda_G$ or conclude that $\Lambda_G=\emptyset$, where the maximum vector exists due to semi-lattice property of $\Lambda_G$. 
The lemma below ensures that it is possible in polynomial time. 

\begin{lemma}\label{lemma:FPT_main-simple}
For a star-forest $G$, let $\Lambda=\prod_{j\in[t]}\Lambda_j$ be a non-empty set such that $\Lambda \supseteq \Lambda_G$. 
If the maximum vector  $\lambda=(\max \Lambda_j)_{j\in[t]}$ does not belong to $\Lambda_G$, then an index $\ell\in[t]$ with $(\Lambda_{\ell}-\max\Lambda_{\ell}) \times \prod_{j\in[t]-\ell}\Lambda_j \supseteq \Lambda_G$ can be computed in polynomial time. 
\end{lemma}
Let us note that an index $\ell$ in the lemma must exist again by the semi-lattice property of $\Lambda_G$. 
\begin{comment}
\begin{proof}[Proof Sketch]
Split property (i.e., Lemma \ref{lem:starforest}) implies that $G$ and $\lambda$ satisfies 
(\ref{eq:splitq1}) and (\ref{eq:splitq2}).
Moreover, by Lemma \ref{lemma-monotonefeasible} in the appendix, we show that feasibility and monotonicity 
imply (\ref{eq:splitq3}), (\ref{eq:splitq4}),  and (\ref{eq:splitq5}). 
If one of the conditions is violated, then a desired index $\ell$ can be computed in polynomial time. 
For example, any index $\ell$ satisfies 
(\ref{eq-lattice000}) if (\ref{eq:splitq1}) or (\ref{eq:splitq4}) is violated. 
If  
 (\ref{eq:splitq2}) or  (\ref{eq:splitq3}) is violated for an index $j$, then 
$\ell=j$ satisfies 
(\ref{eq-lattice000}). 
If (\ref{eq:splitq3}) is violated, i.e, $\lambda_{j-1} < \lambda_j$ holds for some $j$, then $\ell=j$ again satisfies (\ref{eq-lattice000}).

Finally, if  a feasible allocation $\cT$ provided by $G$ and $\lambda$ violates envy-freeness, i.e., 
 $a \in T \,(\in \cT_h)$ envies $a' \in T'\,(\in \cT_j)$ for some $j,h \in [t]$. 
 If $j=h$, 
 then it is clear that $\ell=j\,(=h)$ satisfies 
(\ref{eq-lattice000}).
On the other hand, if $j\not= h$, 
Let $\lambda'$ be a size vector in $\Lambda$ such that $\lambda'_\ell=\lambda_\ell$ and satisfies (\ref{eq:splitq2}), (\ref{eq:splitq3}), (\ref{eq:splitq4}),  and (\ref{eq:splitq5}). 
Then $a$ still envies $a'$ in the allocation provided by $G$ and $\lambda'$. 
Thus $\ell=j$ again satisfies (\ref{eq-lattice000}).
Since envy-freeness can be checked in polynomial time, 
this completes the proof. 
\end{proof}
\end{comment}
Let $\Lambda=\prod_{j\in[t]}\Lambda_j$ denote a set of candidate size vectors. 
By Lemma \ref{lem:starforest}, we have $\Lambda_G \subseteq \prod_{j\in [t]}\bigl[|A_{\in C_j}|\bigr]$. 
Our algorithm initializes $\Lambda$ by $\Lambda=\prod_{j \in [t]}\bigl[|A_{\in C_j}|\bigr]$, and iteratively check if $\Lambda=\emptyset$ or the maximum vector in $\Lambda$ provides an envy-free allocation; 
If not, it updates $\Lambda$ by making use of indices $\ell$ in Lemma \ref{lemma:FPT_main-simple}, where the formal description of the algorithm can be found in Algorithm \ref{alg:FPT_types}.

\begin{theorem}\label{thm:FPT:envyfree:type}
We can check the existence of an envy-free feasible allocation, and find one if it exists in FPT with respect to the number of types of agents. 
\end{theorem}
\begin{proof}
We show that Algorithm~\ref{alg:FPT_types} can check the existence of an envy-free feasible allocation and find one if it exists in FPT time. The correctness follows from Lemmas \ref{lem:starforest},  \ref{lemma-monotonefeasible}, and \ref{lemma:FPT_main-simple}. 
To analyze the running time, observe that the number of iterations of the while loop is at most $n$ because $\sum_{j\in[t]}|\Lambda_j|=n$ at the beginning of the loop  and it is decremented by at least one in each iteration. The running time of each iteration of the while loop is $O(n^3)$ because we can check the existence of envy in $O(n^3)$ time. Thus, the total running time of the algorithm is $O(p^p\cdot n^4)$, which is FPT.
\end{proof}

\begin{algorithm}
\caption{FPT w.r.t.\ the number of types}\label{alg:FPT_types}
\ForEach{star-forest $G$}{
  Let $\Lambda=\prod_{j\in[t]}\bigl[|A_{\in C_j}|\bigr]$\;
  \While{$\Lambda\ne\emptyset$}{
    Let $\lambda=(\max \Lambda_j)_{j\in[t]}$\;
    \If{(\ref{eq:splitq1}) or (\ref{eq:splitq4}) is violated}{Set $\Lambda \ot \emptyset$\;} 
    \ElseIf{ (\ref{eq:splitq2}), (\ref{eq:splitq3}),  or (\ref{eq:splitq5}) is violated for an index $j$}{Set $\Lambda_j \ot \Lambda_j -\max\Lambda_j$\;}
    \ElseIf{an allocation $\cT$ provided by $G$ and $\lambda$ is not envy-free, i.e., an agent in some coalition in $\cT_j$ is envied}{Set $\Lambda_j \ot \Lambda_j -\max\Lambda_j$\;}
    \Else{\Return an allocation $\cT$ provided by $G$ and $\lambda$\;}   
  }
}
\Return ``No envy-free feasible allocation''\;
\end{algorithm}

\subsection{Consecutive envy-free allocations}\label{sec:consecutive}
One desirable property of an allocation is \emph{consecutiveness}, i.e., coalitions are formed by consecutive agents  according to their destinations. The property is intuitive to the users and hence important in practical implementation. Formally, an allocation $\cT$ is \emph{consecutive} if $\max_{a\in T}x_a\le \min_{a\in T'}x_a$ or $\min_{a\in T}x_a\ge \max_{a\in T'}x_a$ holds for all distinct $T,T'\in\cT$. 
However, there exists an instance with no consecutive envy-free feasible allocation as illulstrated in Example \ref{ex:consecutive}. 

\begin{example}\label{ex:consecutive}
Consider an instance where $n=10$, $k=2$, $q_1=6$, $q_2=4$,
$x_1=\dots=x_4=1$, $x_5=\dots=x_8=10$, and $x_9=x_{10}=20$.
Then, it can be easily checked that allocation $\cT^*=(\{1,2,3,4,9,10\},\{5,6,7,8\})$ in Fig.~\ref{fig:no-consecutive} is envy-free and feasible but not consecutive.
Moreover, this is a unique allocation that is  envy-free and feasible. 
To see this, let $\cT=(T_1,T_2)$ be an envy-free feasible allocation. By feasibility, we have $|T_1|=6$ and $|T_2|=4$. We also have $\{1,2,3,4\} \subseteq T_1$, i.e., all agents of type $x=1$ must be allocated to $T_1$ since the cost they have to pay at $T_1$ and $T_2$ are respectively  $\frac{1}{6}$ and $\frac{1}{4}$. 
Finally, all agents of type $x=10$ must be allocated to $T_2$, which completes the uniqueness of $\cT^*$. 
Note that feasibility implies at least two agents of type $x=10$ allocated to $T_2$. 
If an agent of type $x=10$ is allocated to $T_1$, then she would envy an agent of the same type allocated to $T_2$ since the cost she has to pay at $T_1$ is $\frac{1}{6}+\frac{9}{2}$, which is greater than the cost of $\frac{10}{4}$ at $T_2$. 
%Also, if $\cT$ splits the agents of type $x=10$ or the agents of type $x=20$, this would contradict Lemma \ref{lemma:split} $($i$)$ because such agents will not be the first passengers to drop off in $T_1$. Thus, the only possible case is when $\cT=\cT^*$. This, in turn, implies that no envy-free feasible allocation is consecutive. 
%because agents $1,2,3,4$ must be allocated to $T_1$ (otherwise they envy agents in $T_1$) and thus $5,6,7,8$ must be allocated to $T_2$ (otherwise they envy agents in $T_2$).
\end{example}

\begin{figure}
\centering
\begin{tikzpicture}[xscale=.2]
\draw[|->,thick] (0,1)--(22,1);
\draw[|->,thick] (0,0)--(22,0);
\node[left] at (0,1) {$T_1$};
\node[left] at (0,0) {$T_2$};
\node[fill=black, shape=circle, inner sep=2pt, label=above:{1,2,3,4}, label={[font=\tiny,blue]below:$1$}] at (1,1) {};
\node[fill=black, shape=circle, inner sep=2pt, label=above:{5,6,7,8}, label={[font=\tiny,blue]below:$10$}] at (10,0) {};
\node[fill=black, shape=circle, inner sep=2pt, label=above:{9,10}, label={[font=\tiny,blue]below:$20$}] at (20,1) {};
\end{tikzpicture}
\caption{An allocation that is envy-free but not consecutive}\label{fig:no-consecutive}
\end{figure}
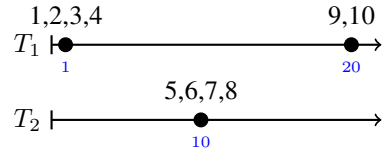

Nevertheless, we show next that a consecutive envy-free feasible allocation can be found in polynomial time if it exists. The key observation here is that  envy-freeness (between all agents) is equivalent to envy-freeness between two consecutive agents, which enables us to design a dynamic programming approach for finding a desired allocation. 

\begin{theorem}\label{thm:envyfree:consecutive}
A consecutive envy-free feasible allocation can be computed in polynomial time  if it exists.
\end{theorem}

%Let us consider a consecutive feasible allocation defined by 
By the feasibility of allocations,  monotonicity property of envy-freeness, and the assumption that $q_1\ge\dots\ge q_k$,
it is sufficient to consider consecutive allocations $\cT=(T_1, \dots ,T_h)$ with $h \leq k$ such that
\begin{align}
T_i=\{s_i,s_{i}+1,\dots,t_i\}\ \ \mbox{ for all } i\in[h],
\label{eq:cons1}
\end{align}
where $s_i$ and $t_i$ are positive integers with 
\begin{align*}
&s_1=1<  s_2=t_1+1 < \dots < s_h=t_{h-1}+1 \leq t_h=n,\\ 
&t_i-s_i+1\leq q_i \ \ \mbox{ for all }i \in [h], \mbox{ and}\\
&t_1-s_1 \geq t_2-s_2 \geq \dots \geq t_h-s_h \geq 0, 
\end{align*}
where the second and third conditions follow from capacity condition and monotonicity property, respectively.  
We regard a partition with  (\ref{eq:cons1})  as a partition $\{T_1, \dots, T_k\}$ satisfying (\ref{eq:cons1}) and the condition that $T_i =\emptyset$ for all taxis $i$ with $h < i \leq k$

We first show a simple criterion on envy-freeness for consecutive allocations.  
\begin{lemma}\label{lemma:cons_check}
A consecutive allocation $\cT$ of \eqref{eq:cons1} is envy-free if and only if for every $i\in [h-1]$, $t_i$ and $s_{i+1}$ do not envy each other.  
\end{lemma}
\begin{proof}
Since the ``only if'' part is clear, we prove the ``if'' part.
Suppose that  $a_i\in T_i$ is the minimum agent that envies an agent   $a_j\in T_j$. 
Since $i\not=j$,  we separately consider two cases $i<j$ and  $i>j$.
\smallskip

\noindent\textbf{Case of $i<j$}.
As $a_i$ envies $a_j$, we have
\begin{align*}
\phi(T_i,x_{a_i})&=\int_{0}^{x_{a_i}}\frac{\d{r}}{n_{T_i}(r)}\\
&  >\phi(T_j-a_j+a_i,x_{a_i})=\frac{x_{a_i}}{|T_j|} \geq \frac{x_{a_i}}{|T_{i+1}|}, 
\end{align*}
where the last inequality follows from the monotonicity. 
Note that $\frac{1}{x}\int_{0}^{x}\frac{\d{r}}{n_{T_i}(r)}$ is monotone nondecreasing in $x$,
since $n_{T_i}(r)$ is monotone nonincreasing in $r$.
Hence, we have
\begin{align*}
\frac{1}{x_{t_i}}\int_{0}^{x_{t_i}}\frac{\d{r}}{n_{T_i}(r)}
\ge \frac{1}{x_{a_i}}\int_{0}^{x_{a_i}}\frac{\d{r}}{n_{T_i}(r)}
> \frac{1}{|T_{i+1}|}.
\end{align*}
Thus, we obtain
\begin{align*}
\phi(T_i,x_{t_i})
&=\int_{0}^{x_{t_i}}\frac{\d{r}}{n_{T_i}(r)}\\
%&=\phi(T_i,x_{a_i})+\int_{x_{a_i}}^{x_{t_i}}\frac{\d{r}}{n_{T_i}(r)}\\
%&>\frac{x_{a_i}}{|T_{i+1}|}+\int_{x_{a_i}}^{x_{t_i}}\frac{\d{r}}{n_{T_i}(r)}
%\ge \frac{x_{t_i}}{|T_{i+1}|}
&>\frac{x_{t_i}}{|T_{i+1}|}=\phi(T_{i+1}-s_{i+1}+t_{i},x_{t_i}),
\end{align*}
meaning that $t_i$ envies $s_{i+1}$.

\smallskip
\noindent\textbf{Case of $i>j$}. As $a_i$ envies $a_j$, we have
\begin{align*}
\phi(T_i,x_{a_i})
&>\phi(T_j-a_j+a_i,x_{a_i})\\
&=\phi(T_j-a_j+t_{i-1},x_{t_{i-1}})+(x_{a_i}-x_{t_{i-1}})\\
&\ge\phi(T_{i-1},x_{t_{i-1}})+(x_{a_i}-x_{t_{i-1}})\\
&=\phi(T_{i-1}-t_{i-1}+a_i,x_{t_{i-1}})+(x_{a_i}-x_{t_{i-1}})\\
&=\phi(T_{i-1}-t_{i-1}+a_i,x_{a_{i}}),
\end{align*}
where the second inequality holds since $t_{i-1}$ never envies $a_j$ by the minimality of $a_i$.
Hence, we have
\begin{align*}
\phi(T_i,x_{s_i})
&=\phi(T_i,x_{a_i})-\int_{x_{s_i}}^{x_{a_i}}\frac{\d{r}}{n_{T_i}(r)}\\
&>\phi(T_{i-1}-t_{i-1}+a_i,x_{a_{i}})-\int_{x_{s_i}}^{x_{a_i}}\frac{\d{r}}{n_{T_i}(r)}\\
&\ge \phi(T_{i-1}-t_{i-1}+a_i,x_{a_{i}})-(a_i-x_{s_i})\\
&=\phi(T_{i-1}-t_{i-1}+s_i,x_{s_{i}}),
\end{align*}
meaning that $s_{i}$ envies $t_{i-1}$. 
\end{proof}

\begin{proof}[Proof of Theorem \ref{thm:envyfree:consecutive}]
For positive integer $\mu$ and $\kappa$ with  $\mu\leq n$ and $\kappa \leq k$, let us consider the subproblem in which $[\mu]$ is the set of agents and $[\kappa]$ is the set of taxis.  
For a nonnegative integer $\ell$,  
let  $z(\mu,\kappa,\ell)$ be a mapping to $\{0,1\}$ such that  $z(\mu,\kappa,\ell)=1$ if and only if 
the subproblem has  a consecutive envy-free feasible allocation $\cT$  with $|T_{\kappa}|=\ell$. 
When there is only one taxi (i.e., $\kappa=1$), it is not difficult to see that
\begin{align}
z(\mu,1,\ell)&=\begin{cases}
1&\text{if }\mu=\ell\text{ and }\mu\le q_1,\\
0&\text{otherwise}.
\end{cases}
\end{align}
Moreover, by Lemmas \ref{lemma:head_size} and \ref{lemma:cons_check},  for an integer  $\kappa$ with $1 < \kappa \leq k$, 
we have $z(\mu,\kappa,\ell)=1$ if and only if
there exists $\ell'\in [n]$ such that 
\begin{align*}
&\ell \leq \ell' \leq q_{\kappa-1}, \\
&z(\mu-\ell,\kappa-1,\ell')=1,\\
&\phi(T',x_{\mu-\ell})\le\phi(T\setminus\{\mu-\ell+1\}\cup\{\mu-\ell\},x_{\mu-\ell}), \mbox{ and}\\
&\phi(T,x_{\mu-\ell+1})\le\phi(T'\setminus\{\mu-\ell\}\cup\{\mu-\ell+1\},x_{\mu-\ell+1}),
\end{align*}
where $T'=\{\mu-\ell-\ell'+1,\dots,\mu-\ell\}$ and $T=\{\mu-\ell+1,\dots,\mu\}$.
Therefore, the original instance contains a consecutive envy-free feasible allocation if and only if $\max_{\ell\in[n]}z(n,k,\ell)=1$. 
If this is the case, such an allocation can be found using  a standard dynamic programming approach, which requires $O(k n^3)$ time. 
\end{proof}

\subsection{Hardness results}
%In this section, we discuss the computational hardness of obtaining an envy-free feasible allocation. 
Having established polynomial-time algorithms for several cases, we turn our attention to the general problem of computing an envy-free feasible allocation.  
Unfortunately, it remains an open question whether the problem of deciding the existence of an envy-free feasible allocation is NP-hard or polynomial-time solvable. 
%We instead prove the NP-hardness of deciding the existence of a feasible allocation that satisfies a necessary condition to be envy-free.
We instead consider two natural relaxations of envy-freeness and prove the NP-hardness of deciding the existence of such allocations.
%One is a relaxed envy-free feasible allocation, and the other is a generalized envy-free feasible allocation. 

%The relaxed envy-freeness means that it satisfies a necessary condition to be envy-free. 
The first one relaxes the envy-free requirement, by imposing the necessary conditions in Split Lemma. More precisely, we consider the conditions (i)--(iii) in Lemma~\ref{lemma:split}. 
We say that a feasible allocation $\cT$ satisfies \emph{the split conditions} if the conditions (i)--(iii) in Lemma~\ref{lemma:split} are satisfied for any $x$ and distinct $T,T'\in\cT$ such that $T_{=x}$ and $T'_{=x}$ are non-empty. Computing such an allocation turns out to be NP-hard. 

\begin{figure}
\centering
\begin{tikzpicture}[xscale=.5]
\draw[|->,thick] (0,3)--(14,3);
\draw[|->,thick] (0,2)--(14,2);
\draw[|->,thick] (0,0)--(14,0);
\node[left] at (0,3) {$T_1$};
\node[left] at (0,2) {$T_2$};
\node[left] at (0,1) {$\vdots$};
\node[left] at (0,0) {$T_m$};
\node[fill=black, shape=circle, inner sep=2pt, label={[font=\scriptsize]above:{$1$}}, label={[font=\tiny,blue]below:{}}] at (1,3) {};
\node[fill=black, shape=circle, inner sep=2pt, label={[font=\scriptsize]above:{$2$}}, label={[font=\tiny,blue]below:{}}] at (1,2) {};
\node[fill=black, shape=circle, inner sep=2pt, label={[font=\scriptsize]above:{$m$}}, label={[font=\tiny,blue]below:{1}}] at (1,0) {};

\node[fill=black, shape=circle, inner sep=2pt, label={[font=\scriptsize]above:{$\beta\cdot a_{i_1^1}$}}, label={[font=\tiny,blue]below:{}}] at (2,3) {};
\node[fill=black, shape=circle, inner sep=2pt, label={[font=\scriptsize]above:{$\beta\cdot a_{i_2^1}$}}, label={[font=\tiny,blue]below:{}}] at (5,3) {};
\node[fill=black, shape=circle, inner sep=2pt, label={[font=\scriptsize]above:{$\beta\cdot a_{i_3^1}$}}, label={[font=\tiny,blue]below:{}}] at (7,3) {};

\node[fill=black, shape=circle, inner sep=2pt, label={[font=\scriptsize]above:{$\beta\cdot a_{i_1^2}$}}, label={[font=\tiny,blue]below:{}}] at (4,2) {};
\node[fill=black, shape=circle, inner sep=2pt, label={[font=\scriptsize]above:{$\beta\cdot a_{i_2^2}$}}, label={[font=\tiny,blue]below:{}}] at (6.5,2) {};
\node[fill=black, shape=circle, inner sep=2pt, label={[font=\scriptsize]above:{$\beta\cdot a_{i_3^2}$}}, label={[font=\tiny,blue]below:{}}] at (8,2) {};

\node[fill=black, shape=circle, inner sep=2pt, label={[font=\scriptsize]above:{$\beta\cdot a_{i_1^m}$}}, label={[font=\tiny,blue]below:{}}] at (3,0) {};
\node[fill=black, shape=circle, inner sep=2pt, label={[font=\scriptsize]above:{$\beta\cdot a_{i_2^m}$}}, label={[font=\tiny,blue]below:{}}] at (5,0) {};
\node[fill=black, shape=circle, inner sep=2pt, label={[font=\scriptsize]above:{$\beta\cdot a_{i_3^m}$}}, label={[font=\tiny,blue]below:{}}] at (7.5,0) {};

\node[fill=black, shape=circle, inner sep=2pt, label={[font=\scriptsize]above:{$(S+1)\beta-1$}}, label={[font=\tiny,blue]below:{$3m+2$}}] at (10,3) {};
\node[fill=black, shape=circle, inner sep=2pt, label={[font=\scriptsize]above:{$(S+1)\beta-2$}}, label={[font=\tiny,blue]below:{$3m+3$}}] at (11,2) {};
\node[fill=black, shape=circle, inner sep=2pt, label={[font=\scriptsize]above:{$(S+1)\beta-m$}}, label={[font=\tiny,blue]below:{$4m+1$}}] at (13,0) {};
\end{tikzpicture}
\caption{A feasible allocation that satisfies the split conditions}\label{fig:hard}
\end{figure}
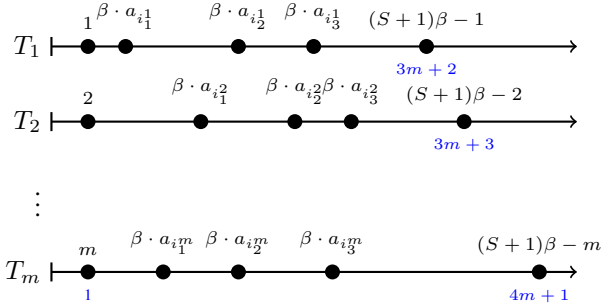

\begin{theorem}\label{thm:hard_split}
It is NP-complete to decide whether there exists a feasible allocation that satisfies the split conditions.
\end{theorem}
\begin{proof}
We provide a reduction from the \emph{3-partition} problem, which is a strongly NP-complete problem~\citep{GJ1979}.
In the problem, we are given $3m+1$ positive integers $a_1,a_2,\dots,a_{3m},S$ satisfying $S/4<a_i<S/2~(\forall i\in [3m])$ and $\sum_{i\in[3m]}a_i=mS$.
Our task is to decide whether there exists a partition $(I_1,\dots,I_m)$ of the index set $[3m]$ such that $\sum_{i\in I_j}a_i=S$ for any $j\in[m]$.
Note that, by the condition $S/4<a_i<S/2~(\forall i\in [3m])$, every such $I_j$ must contain exactly three elements from $[3m]$.

Let $a_1,a_2,\dots,a_{3m},S$ be an instance of the 3-partition problem.
We construct a ride allocation instance $(A,[k],(x_a)_{a\in A},(q_i)_{i\in [k]})$ which has a feasible allocation that satisfies the split conditions if and only if the given 3-partition instance is a Yes-instance.
Let $\beta=m(m+1)$.
We set the number of taxis $k$ to be $m$ and the capacity of each taxi to be $q=(2S+1)\beta$.
The agents $A$ are partitioned into $4m+1$ groups by the destination types $\{1,2,\dots,4m+1\}$.
We set the number of agents of type $1$ to be $m(m+1)/2~(=1+2+\dots+m)$. 
We will see that the agents of type $1$ must be the first passengers to drop off in every taxi.
The following $3m$ types $P\coloneqq\{1+i\mid i\in [3m]\}$ are associated with the index set $[3m]$.
For each $i\in[3m]$, we set the number of agents of type $i+1$ to be $\beta\cdot a_i$.
The remaining types $D\coloneqq\{3m+1+i\mid i\in [m]\}$ are dummy to ensure that the agents of type $1$ are split into all the taxis.
For each $i\in [m]$, we set the number of agents of type $3m+1+i$ to be $(S+1)\beta-i$.
Note that the total capacity of the taxis and the number of agents are both $(2S+1)m\beta$, and hence we must allocate $q$ agents for each taxi.

Suppose that the given 3-partition instance is a Yes-instance, i.e., 
there exists a partition $(I_1,\dots,I_m)$ of the index set $[3m]$ such that $\sum_{i\in I_j}a_i=S$ for any $j\in[m]$.
Let $I_j=\{i^j_1,i^j_2,i^j_3\}$ and let $(H_1,\dots,H_m)$ be a partition of $A_{=1}$ such that $|H_j|=j$ for each $j\in [m]$.
Let $\cT$ be an allocation with $T_j=H_j\cup \bigcup_{i\in I_j}A_{=1+i}\cup A_{=3m+1+j}$.
Then, it is not difficult to see that $\cT$ is feasible and satisfies the split conditions (see Fig.~\ref{fig:hard}).

Conversely, suppose that there exists a feasible allocation $\cT$ that satisfies the split conditions.
We first show that there is at least one agent of type $1$ in every taxi, i.e., $T_i\cap A_{=1}\ne\emptyset$.
Let $J$ be the set of taxis which contains some agent of type $1$, i.e., $J=\{i\in[k]\mid T_i\cap A_{=1}\ne\emptyset\}$.
Then, by the condition (i), $A_{=x}\subseteq T_i$ or $A_{=x}\cap T_i=\emptyset$ for any $x>1$ and $i\in J$.
Let $Q$ be the set of types of which the agents ride a taxi in $J$, i.e., $Q=\{x>1\mid A_{=x}\subseteq T_i~(\exists i\in J)\}$.
Since $q$ is a multiple of $\beta$, the number of agents who ride a taxi in $J$ is also a multiple of $\beta$.
By counting the number of agents modulo $\beta$, we obtain
\begin{align*}
0
&\equiv |A_{=1}|+\sum_{x\in Q}|A_{=x}| \pmod{\beta}\\
&\equiv \frac{m(m+1)}{2}+\sum_{i\in[m]:\,3m+1+i\in Q}|A_{=x}| \pmod{\beta}\\
&\equiv\frac{m(m+1)}{2}-\sum_{i\in[m]:\,3m+1+i\in Q}i \pmod{\beta}
\end{align*}
Thus, $Q$ must contain all the types in the dummy types $D$ (recall that $\beta>\frac{m(m+1)}{2}$).
Here, each taxi in $J$ cannot carry two type in $D$ because the number of agents of each type in $D$ is larger than the half of the capacity of each taxi $q$.
Hence, we conclude that there is at least one agent of type $1$ in every taxi, i.e., $J=[k]$.

Now, we prove that there exists a desired partition of $[3m]$.
Without loss of generality, we may assume that $T_i$ contains the agents of type $3m+1+i$ for each $i\in[m]$.
Since $q$ is a multiple of $\beta$ and the number of agents of type $x\in P$ is a multiple of $\beta$, the number of agents of type $1$ in the $i$th taxi must be $i$, i.e., $|T_i\cap A_{=1}|=i$.
Let $Q_i$ be the set of types of which the agents ride $i$th taxi, i.e., $Q_i=\{x\in P\mid A_{=x}\subseteq T_i\}$.
Then, we have $\sum_{x\in Q_i}|A_{=x}|=q-((S+1)\beta-i)-i=S\cdot\beta$, and hence the partition $(I_1,\dots,I_k)$ of $[3m]$ with $I_j=\{i\in[3m]\mid i+1\in Q_j\}$ satisfies $\sum_{i\in I_j}a_i=S$ for any $j\in[m]$.
\end{proof}

The second relaxation generalizes the notion of envy-freeness, by only looking into envies within particular groups. 
%The latter is a generalized envy-free feasible allocation. 
%For a set $S$ of agents, we consider an allocation in which the agents in $S$ never envy each other. We call such an allocation \emph{envy-free in $S$}. Such an allocation is desirable in a case when the members in $S$ are friends and know each other.
%${\cal S}$ forms a partition of $A$, that is, $\bigcup_{S\in cal S}S=A$ and $S_i\cap S_j=\emptyset$ for $i\neq j$.
%The notion can be further extended to multiple sets of agents, ${\cal S}=\{S_1,S_2,\ldots,S_q\}$, and we say 
For multiple sets of agents ${\cal S}=\{S_1,S_2,\ldots,S_q\}$, we say that an allocation is \emph{envy-free in ${\cal S}$} if for each $S\in {\cal S}$, the agents in $S$ do not envy each other. %A typical setting is that ${\cal S}$ forms a partition of $A$, where $A=\bigcup_{S\in{\cal S}}S$ and $S_i \cap S_j =\emptyset$ for any distinct $S_i, S_j\in {\cal S}$. 
The notion of \emph{envy-freeness in ${\cal S}$} is a generalized envy-freeness in the sense that ${\cal S}=\{A\}$ coincides with the normal envy-freeness. 
Such a generalized envy-freeness is useful to control the rank-wise service quality.  
For example, in a frequent flyer program of a airline company, agents in an identical status are supposed to receive a similar quality of services. In the context of our problem, it is desirable that agents in ${\cal S}$, a set of frequent flyers in a status, never envies another agent in ${\cal S}$. 
%Another motivation to consider such an envy-freeness is the situation where the each agent is classified into a class that represents a service rank, such as a frequent flyer status of an airline company. In  
Unfortunately, it is also NP-hard to find an allocation that is envy-free in a given ${\cal S}$.  
 
\begin{theorem}\label{thm:hard_envyfree_inS}
Given a partition ${\cal S}$ of $A$, it is NP-complete to decide whether there exists a feasible allocation that is envy-free in ${\cal S}$. 
\end{theorem}
\begin{proof}
We provide a reduction from the \emph{Numerical 4-dimensional matching} (N4DM) problem, which is a variant of \emph{4-partition}. 
In an N4DM instance, we are given a positive integer $p$ and four sets of $k$ positive integers  $S_a=\{a_1,a_2,\ldots,a_{k}\}$, $S_b=\{b_1,b_2,\ldots, b_k\}$,  $S_c=\{c_1,c_2\ldots,c_k\}$ and $S_d=\{d_1,d_2\ldots,d_k\}$. Here, we can impose another condition that all the numbers in $S_a \cup S_b \cup S_c \cup S_d$ are distinct. 
Our task is to decide whether there exists a subset $M$ of $S_a\times S_b\times S_c\times S_d$ such that every integer in $S_a$, $S_b$,  $S_c$  and $S_d$ occurs exactly once and that for every quadruple $(a,b,c,d)\in M$ $a+b+c+d=p$ holds. 
%The triple version N3DM is labeled as [SP16] in \citet{GJ1979}, 
The hardness of 4-partition is shown in \citet[Theorem 4.3]{GJ1979}. 
It actually proves N4DM with the distinct condition.  
We can further assume without loss of generality that  $3\max S_a < \min S_b$, $2\max S_b < \min S_c$ and $2\max S_c < \min S_d$, because otherwise we can use  $S'_b=\{b+n^{\alpha} \mid b\in S_b\}$ with a large constant $\alpha$ instead of $S_b$, for example. Thus, $\min S_a < \max S_a \ll \min S_b < \max S_b \ll \min S_c < \max S_c \ll \min S_d < \max S_d$ holds roughly. Furthermore, we assume that $k \equiv_{120} 1$, that is, $k=120 k'+1$ for some positive integer $k'$.

The reduction is as follows: We prepare $4k$ agents for $S_a\cup S_b \cup S_c \cup S_d$ together with extra $k$ agents, called $S_e$. Thus we have $5k$ agents in total.  
For $a \in S_a$, $b \in S_b$, $c\in S_c$ and $d\in S_d$, the destinations of the corresponding agents $a,b,c$ and $d$ are respectively $x_a=20a$, $x_b=12b$, $x_c = 6c$ and $x_d=2d$. 
The destination of every extra agents $e_i\in S_e$ is all $x_{e_i}=60p$. The capacities of the $k$ taxis are also same $5$, and thus all the taxis should be full in a feasible allocation. The partition is defined by the types, that is, ${\cal S}=\{S_x \mid x\in \mathbb{R}_{>0}\}$, where $S(x) =\{a\in A\mid x_a =x \}$.  

We claim that the instance has an envy-free allocation in ${\cal S}$ if and only if the distinct N4DM instance is a yes-instance. 
We first show the if direction. We assume $M$ is a yes-solution, that is, any triple $(a,b,c,d)\in M$, $a+b+c+d=p$ holds. For a triple $(a,b,c,d)\in E$, we let $a,b,c,d$ and an $e\in S_e$ take a taxi. 
Note that $x_{a}<x_{b}<x_{c}<x_{d}<x_{e}$. Then their payments are as follows: $x_a/5=4a$ for agent $a$, $x_a/5+(x_b-x_a)4=3b-a$ for agent $b$, $x_a/5+(x_b-x_a)/4+(x_c-X_b)/3=2c-b-a$ for agent $c$, 
$x_a/5+(x_b-x_a)/4+(x_c-x_b)/3+(x_d-x_c)/2=60p-c-b-a=99p$ for agent $d\in S_d$. Since every agent in $S_d (=S(60p))$ pays exactly $99p$, agents in $S(60p)$ never envy each other. We then check the envy-freeness of the agents in $S(x_c)$. 
As seen above, the payment of an agent in $S(x_c)$ is $c-(a+b)=2c-p$, which does not depend on which taxi $c$ takes. 

We next consider only-if direction. Assume that there exists a feasible allocation in which any $d_i \in S$ does not envy another $d_j\in S$. In a feasible allocation, each taxi has exactly one agent in $S$; otherwise, there are two taxis with capacity 4 that deliver different numbers of agents in $S$ due to $k\equiv_{24}=\equiv_{2\cdot 3\cdot 4} 1$, which makes an envy. For example, suppose that a taxi has 3 agents in $S$ and a taxi has 4 agents in $S$. Then, an agent in the former taxi envies one in the latter taxi, because an agent in the former taxi pays $60p/3-\max{S_3}/4>30p$ and an agent in the latter taxi pays $60p/4=15p$. 

Thus, each taxi has exactly one agent in $S$, whose payment is determined by the other members in the taxi. 
If some taxi has two or more agents from $S_2$ and another taxi has at most one $S_1$, the former taxi is cheaper for the agent in $S$. This and similar arguments imply that every taxi has one agent from each of $S_1$, $S_2$, $S_3$, $S_4$ and $S_5$ in an envy-free feasible allocation. By the argument of if-direction, if a taxi has agent $a$ from $S_1$, agent $b$ from $S_2$, agent $c$ from $S_3$, and agent $d$ from $S_4$, the payment of the remaining $e$ is $60p-(a+b+c+d)$. This implies that such an allocation of agents corresponds to an N4DM solution. 
\end{proof}

We note that, the above proof implies the NP-hardness for another relaxed variant: that is, given a subset $S$ of $A$, it is NP-complete to decide whether there is a feasible allocation that is envy-free in $S$.

\section{Stable and socially optimal allocations}
We have seen that the set of envy-free allocations may be empty even when a feasible outcome exists. In contrast, we will show in this section that, stability as well as social optimality are possible to achieve simultaneously: a feasible allocation that greedily groups agents from the furthest destinations together satisfies Nash stability, strong swap-stability, and social optimality. Specifically, we  design the following \emph{backward greedy} algorithm which constructs coalitions $T_i$ in the increasing order of $i$ by greedily adding agents $j$ in the decreasing order until $T_i$ exceeds the capacity, where the formal description can be found in Algorithm \ref{alg:backward}. %, where the formal description can be found in the supplementary material. 

\begin{algorithm}
\caption{Backward greedy}\label{alg:backward}
Initialize $T_i \ot \emptyset$ for each $i\in [k]$ and let $\kappa \ot 1$\;
\For{$a\ot n$ to $1$}{
 \If{$|T_{\kappa}|=q_{\kappa}$}{    $\kappa\ot \kappa+1$\;
    \lIf{$\kappa>k$}{\Return ``No feasible allocation''}
 }
 Set $T_{\kappa}\ot T_{\kappa}+a$\;
}
%Set $T_{\ell}\ot \emptyset$ for all $\ell$ with $k'+1 \leq \ell \leq k$\;
\Return $(T_1,T_2,\dots,T_{k})$\;
\end{algorithm} 

The following theorem states that Algorithm~\ref{alg:backward} computes a desired outcome in polynomial time.

\begin{theorem}\label{thm:stable:SO}
If a given instance has a feasible allocation, the backward greedy computes in polynomial time a feasible allocation that is socially optimal, Nash stable, and strongly swap stable.
\end{theorem}
\begin{proof}%[Full Proof of Theorem \ref{thm:stable:SO}]
It is not difficult to see that the backward greedy 
 given in Algorithm~\ref{alg:backward} requires $O(n+k)$ time, and computes a feasible allocation if there exists such an allocation.
Let $\cT=(T_1,\ldots,T_{k})$ be a feasible allocation constructed by the algorithm, and let $T_h$ be the last nonempty coalition in $\cT$, i.e., $T_\kappa=\emptyset$ for $\kappa > h$. 
We first that allocation $\cT$ is  Nash stable. 
Note that coalitions $T_1,\dots,T_{h-1}$ have no seat available, and  empty taxis $\kappa~(>h)$ are not profitable to deviate.  
Thus  it is enough to consider deviations to the last coalition $T_{h}$. 
Moreover, 
 if agent $a \in T_i$ wants to deviate to $T_{h}$, 
 then she would become the last passenger to drop off but  $|T_{h}|+1 \leq q_{h} \leq q_i =|T_i|$ holds. 
Thus, by letting $x_b=\max_{t \in T_{h}}x_{t}$,  we have 
\begin{align*}
&\phi(T_{h}+a,x_a) - \phi(T_i,x_a) \\
&\ge (\phi(T_{h}+a,x_b)+(x_a-x_b))-(\phi(T_i,x_b)+(x_a-x_b)),\\
&=\phi(T_{h}+a,x_b)-\phi(T_i,x_b)\\
&\ge \frac{x_b}{|T_{h}|+1}-\frac{x_b}{|T_i|} \geq 0, 
\end{align*}
which yields a contradiction. Thus $\cT$ is Nash stable. 

%%%%%%%%%%%%%%%%%
We next show that $\cT$ is strongly swap-stable. 
Since swapping a pair of agents in the same taxi has no effect on the cost, it suffices to show that there is no beneficial swap of agents in different taxis. 
More precisely,  
let $a \in T_i$ and $b \in T_j$ be two agents with $i <j$. 
We show that  if agent $a$ can replace $b$, i.e.,
\begin{equation}
\label{eq-aplastsec1}
\phi(T_j-b+a,x_{a})\le\phi(T_i,x_{a}), 
\end{equation}
then swapping $a$ and $b$ have no effect on their costs, i.e., 
\begin{align*}
&\phi(T_j-b+a,x_{a})=\phi(T_i,x_{a}),  \mbox{ and} \\
&\phi(T_i-a+b,x_{b})=\phi(T_j,x_{b}). 
\end{align*}

To see this, we first observe that by construction of $\cT$, $|T_i|\ge |T_j|$ and $x_{a}\ge x_t$ for all $t\in T_j$.
Thus,  we have 
\begin{equation}
\label{eq-lastmy1}
n_{T_j-b+a}(x)\le n_{T_i}(x) \ \mbox{ for all } x\in\mathbb{R}_{\ge 0},
\end{equation}
meaning that at any $x>0$, the number of agents in taxi $i$ is at least the number of agents in taxi $j$ with  $a$ and $b$ swapped.
On the other hand, by the definition of $\phi$, (\ref{eq-aplastsec1}) is equivalent to
\begin{align*}
\int_{0}^{x_{a}}\frac{\d{r}}{n_{T_j-b+a}(r)}&\leq  \int_{0}^{x_{a}}\frac{\d{r}}{n_{T_i}(r)}, 
\end{align*}
which together with  (\ref{eq-lastmy1}) implies that  $n_{T_j-b+a}(x)=n_{T_i}(x)$ for all $x\in\mathbb{R}_{\ge 0}$ with $x \leq x_{a}$.
Hence we have $|T_i|=|T_j|$ and $x_t=x_{a}$ for all $t\in T_j$.
This implies that $x_a=x_b$, $n_{T_i}=n_{T_i-a+b}$,  
$n_{T_j}=n_{T_j-b+a}$,  and 
$n_{T_i}(x)=n_{T_j}(x)$ for all $x\in\mathbb{R}_{\ge 0}$ with $x \leq x_{a}$, which proves the claim.

%This proves the claim. In fact,  we have
%\begin{align*}
%\phi(T_j-b+a,x_{a})&=\int_{0}^{x_{a}}\frac{\d{r}}{n_{T_j-b+a}(r)}\\
%&=\int_{0}^{x_{a}}\frac{\d{r}}{n_{T_i}(r)}\\
%&=\phi(T_i,x_{a})
%\end{align*}
%and 
%\begin{align*}
%\phi(T_i-a+b,x_{b}) %&=\int_{0}^{x_{b}}\frac{\d{r}}{n_{T_i-a+b}(r)}\\
%&=\int_{0}^{x_{b}}\frac{\d{r}}{n_{T_j}(r)}\\
%&=\phi(T_j,x_{b}). 
%\end{align*}
%%%%%%%%%%%%%%%%%

It remains to show that $\cT=(T_1,\dots ,T_k)$ is socially optimal i.e., $\cT$ is a feasible allocation that minimizes the total cost among feasible allocations. 
For $i \in [k]$, let $y_i$ denote the furthest destination of $T_i$, i.e.,  $y_i=\max_{a \in T_i} x_a$ if $T_i\not=\emptyset$, and $0$ otherwise. 
Then 
 the total cost of $\cT$ is given by $\sum^{k}_{i=1} y_i$. 
Since the capacities satisfy $q_1\geq \dots \geq q_k$,  we have $y_1 \geq \dots \geq y_k$.
We claim that the nonincreasing sequence of the last drop-off points in socially optimal allocations is unique
and identical to that of the allocation obtained by the backward greedy algorithm. 
This implies that $\cT$ is socially optimal. 

Let $\cT'=(T'_1, \dots , T'_k)$ be a socially optimal feasible allocation, and for $i \in [k]$, let $y'_i$ denote the furthest destination of $T_i'$.
Let $z_1,\dots , z_k$ be a sequence obtained from $y_i'$ \,($i \in [k]$) by sorting them in the nonincreasing order. 
Then our claim is equivalent to the condition that $z_i=y_i$ for all $i \in [k]$. 
For an index $i$, let $U_{i}$ be the coalition in $\cT'$ corresponding to $z_i$. 
Then we note that $z_i$ is the maximum $x_a$ among agents $a$ in $A \setminus (\bigcup_{\ell\in [i-1]} U_{\ell})$. 
Since $\sum_{\ell  \in [i-1]}|U_\ell| \leq \sum_{\ell \in [i-1]} q_\ell$, the backward greedy construction implies  $z_i \geq y_i$. 
Since $\cT'$ is socially optimal, i.e., $\sum_{i \in [k]}z_i=\sum_{i \in [k]}y_i$, we have  $z_i=y_i$ for all $i \in [k]$, which  
 completes the proof. 
\end{proof}

As a  corollary of the above theorem, we can see that there exists a  feasible allocation that satisfies all the notions defined in Section \ref{sec:solution}, except for envy-freeness,  whenever a feasible allocation exists. 
We remark that the backward greedy algorithm fails to find an envy-free feasible allocation even when it exists, since there is an instance that has an envy-free feasible allocation but no consecutive envy-free feasible allocation, which can be found in Example \ref{ex:consecutive}. % of the supplemental material

\section{Conclusion}
In this paper, we introduced a new model of the fair ride allocation problem on a line with an initial point. We proved that the backward greedy allocation satisfies Nash stability, strong swap-stability, and social optimality. We designed several efficient algorithms to compute an envy-free feasible allocation when some parameter of our input is small. The obvious open problem is the complexity of finding an envy-free allocation for the general case. We expect that the problem becomes NP-hard even when the maximum capacity is a constant. 

There are several possible extensions of our model. First, while we have assumed that agents ride at the same starting point, it would be very natural to consider a setting where the riding locations may be different. Indeed, passengers ride at different points in most of private carpooling services. Extending our results to this setting would be a promising research direction. Further, besides the class of path graphs, there are other underlying structures of destinations, such as grids and planar graphs. Although we expect that the Shapley value of a cost allocation problem on a more general graph structure may become necessarily complex, it would be interesting to analyze the properties of fair and stable outcomes in such scenarios. 
%Besides the Shapley value, there are other division rules of cooperative games that can be applied to our model. Examples include the Owen value and the Banzhaf value. It would be interesting to see whether using a different rule of cooperative games results in a different property of an allocation. 

%\newpage
\bibliographystyle{named}

\end{document}